\documentclass[journal,twoside,web]{ieeecolor}
\usepackage{generic}
\usepackage{textcomp}

\def\BibTeX{{\rm B\kern-.05em{\sc i\kern-.025em b}\kern-.08em
    T\kern-.1667em\lower.7ex\hbox{E}\kern-.125emX}}
\usepackage{cite}

\usepackage{graphicx}
\usepackage{amsmath,amssymb,amsfonts}
\usepackage{amsfonts}
\usepackage{amssymb}
\usepackage{mathrsfs}

\usepackage{dsfont}
\usepackage{braket}

\usepackage{algorithm} 
\usepackage{algpseudocode} 
\usepackage{array}
\usepackage[font=small,belowskip=-9pt,aboveskip=4pt,labelsep=period,justification=centering]{caption}
\usepackage[caption=false,font=footnotesize]{subfig}
\usepackage{bigints}

\usepackage{url}

\usepackage{color}

\newcommand{\blue}[1]{{\color{black}#1}}

\newcommand{\mc}[1]{\mathcal{#1}}

\newcommand{\mb}[1]{\mathbb{#1}}
\newcommand{\Ra}{\;\;\Rightarrow\;\;}

\newcommand{\ra}{\rightarrow}

\newcommand{\txt}[1]{\text{#1}}
\newcommand{\tbf}[1]{\textbf{#1}}

\newcommand{\mat}[1]{\begin{matrix}#1\end{matrix}} 
\newcommand{\pa}[1]{\left(#1\right)} 
\newcommand{\br}[1]{\left[#1\right]} 
\newcommand{\pmt}[1]{\pa{\mat{#1}}} 

\newcommand{\q}{\quad}
\newcommand{\s}{&}

\newcommand{\teq}{\triangleq}

\newcommand{\eps}{\epsilon}
\newcommand{\rank}[1]{\txt{rank}\pa{#1}}
\newcommand{\dc}{\text{dc}}

\newtheorem{theorem}{Theorem}

\newtheorem{proposition}{Proposition}
\newtheorem{assumption}{Assumption}

\begin{document}

\title{Approximate Sensitivity Conditioning and Singular Perturbation Analysis for Power Converters}

\author{Kaushik~Gajula,~\IEEEmembership{Student Member,~IEEE,} 
Lalit Kishore Marepalli, \IEEEmembership{Student Member, IEEE,} and Luis~Herrera,~\IEEEmembership{Member,~IEEE}
\thanks{This work was supported by the U.S. National Science Foundation (NSF) under award 1855888. }
\thanks{The authors are with the Department of Electrical Engineering, University at Buffalo, Buffalo, NY 14260, USA  (e-mail: lcherrer@buffalo.edu). }
} 
\maketitle

\begin{abstract}
A   sensitivity conditioning control strategy is analyzed in this paper and it is applied to power electronic converters. The   term is used to improve  closed loop systems, such as power converters with cascaded inner and outer loop controllers. The impact of the   sensitivity term is analyzed  using  singular perturbation theory.  In addition, the  implementation of the   control term is addressed for practical systems, where the number of inputs is generally not sufficient for exact  sensitivity conditioning.  Simulation results are presented for a buck converter with output capacitor voltage regulation and a Permanent Magnet Synchronous Machine (PMSM), used as a generator with an active rectifier. Finally, experimental results are presented for the buck converter, demonstrating the advantages and  feasibility in implementing the approximate sensitivity conditioning term for closed loop power converters. 
\end{abstract}


\section{Introduction}
\IEEEPARstart{P}OWER converters are an essential component of microgrids such as
Plug-in Hybrid Electric Vehicles (PHEVs) \cite{aluisio,aluisio2}, More Electric Aircrafts (MEAs) \cite{giampaolo,xu2019,he2018,kim2018,shirin, 2017HerreraTSG}, electric ships \cite{airan,alfalahi,ali,fang2019,Sulligoi,vu2017}, distribution grids \cite{Zubieta,lvdc,Shenai,Nordman}, etc. Cascaded control is a very popular control strategy for power converters, which consists of PI/PID controllers operating at two different time-scales (inner and outer) \cite{bellinaso2018cascade,liang2019model,pisano2008cascade}. This type of closed loop system can be analyzed using singular perturbation theory, by dividing the states into fast and slow modes \cite{kokotovic}.

Singular perturbation approaches have been used in the analysis and design of closed loop power converters.  One of the main advantages of singular perturbation analysis is the framework it provides for simplifying and reducing the system dynamics, by assuming a sufficiently large time scale separation between slow and fast states. 
In \cite{shen2019singular}, a
dynamic Integrated Energy System (IES) model is proposed based on the singular perturbation theory, considering the dynamic interactions of three energy systems: natural gas, electric power,
and thermal systems. The dynamic IES model, representing a boundary
layer system, is solved and analyzed in three time-scales. In \cite{mchaouar2018nonlinear}, singular perturbation techniques are used for cascaded control of a single phase active filter. In \cite{mchaouar2020new} a synthesized
nonlinear multi-loop controller using singular perturbation technique is proposed, in which the three-time-scale dynamics are artificially induced in the closed-loop system. The proposed nonlinear multi-loop controller is used to control a single-phase grid-connected PV system that includes a full bridge inverter with L-filter. Other applications of singular perturbation analysis and control include dc microgrids with Constant Power Loads (CPL) \cite{gui2021large, liu2021existence, perez2018dc}, motor control \cite{bayardo2020adaptive}, ac microgrids \cite{rasheduzzaman2015reduced, mariani2014model}, power factor correction \cite{kimball2008singular}, etc.

As singular perturbation analysis requires some states to be much faster than the others, sensitivity conditioning control was developed in \cite{picallo2021predictive} as an alternate approach to singular perturbation based high gain feedback \cite{kokotovic}. This sensitivity based conditioning was developed based on recent  algorithms such as gradient trajectory tracking, Newton trajectory tracking, and approximate gradient tracking. In \cite{simonetto2016class}, these algorithms to follow a predictor-corrector scheme to track the solution of a time-varying unconstrained and strongly convex optimization problem. A prediction-correction scheme for solving convex optimization problems with time varying objective and constraint functions was proposed in \cite{fazlyab2017prediction}. Lastly, in \cite{mojica2021stackelberg}, sensitivity conditioning is used to solve a bi-level optimization problem that includes a leader and follower population, obtaining what is known as the Stackelberg equilibrium. \blue{In control systems, the main advantage and motivation for the sensitivity conditioning approach \cite{picallo2021predictive} is to improve the transient response by achieving a strong time-scale separation into fast and slow states, without the need of high gain feedback, which can cause input saturation problems.  However, the actual implementation of the sensitivity conditioning term to practical  systems, where the number of inputs are less than the number of states, has not been analyzed or discussed previously. An exact sensitivity conditioning term cannot be implemented in these systems.}

In this paper, the  sensitivity conditioning control approach presented in \cite{picallo2021predictive} is analyzed from the perspective of singular perturbation theory, i.e. the closed loop system is assumed to compose of slow and fast states, and the impact of the added term is quantified in terms of the time scale separation between the states. The addition of   sensitivity term is analyzed for linear systems, and the closed loop eigenvalues are compared to a case which does not include the sensitivity term. \blue{An approximate sensitivity conditioning term is then proposed to address the implementation of the   term to practical systems where the number of inputs are less than the states}.  Lastly, simulation results are shown for a buck converter with cascaded control and a Permanent Magnet Synchronous Machine (PMSM) used as a generator with active rectification.  Experimental results are also shown for the buck converter, with and without sensitivity conditioning. 

The contributions of this paper are summarized as follows:
\begin{itemize}
    \item The sensitivity conditioning   term proposed in \cite{picallo2021predictive} is analyzed using singular perturbation.  It is shown that as the time scale separation of a system with fast/slow states grows, the impact of the   term is reduced.
    \item Implementation of the   term is discussed for different number of inputs and an approximate sensitivity conditioning is proposed.  The dynamics of systems with approximate sensitivity are simplified and the error of approximate sensitivity conditioning term is also quantified.
    \item Simulation and experimental results are shown for dc/dc converters and a PMSM with active rectification, demonstrating the feasibility of implementation and improvement  of the approximate sensitivity control approach.
\end{itemize}

The paper is organized as follows: in section 2, general singular perturbation analysis for nonlinear and linear systems is presented. In section 3, the  sensitivity control strategy is discussed and its implementation for a linear system is shown in section 4. In section 5, simulation results are presented for case studies, demonstrating the advantage of using approximate sensitivity conditioning strategy for systems with  cascade control. In section 6, experimental results are presented for a closed loop buck converter. Lastly, conclusion and future work are stated in section 7.

\section{Singular Perturbation Overview}

\subsection{Singular Perturbation Modeling}
A nonlinear dynamic system with slow ($x$) and fast $(z)$ states can be described as follows \cite{hkkhalil,kokotovic,o1991singular}:
\begin{align}\label{eq:sing_model}\renewcommand{\arraystretch}{2}
\begin{array}{rl}
     \dfrac{dx}{dt}&=f(x,z),\\
\epsilon \dfrac{dz}{dt}&=g(x,z),
\end{array}
\end{align}
where $x\in\mb{R}^{n_x}$, $z\in\mb{R}^{n_z}$, and $\eps\geq0$ is the singular perturbation parameter modeling the separation of time scales between the two states, where the fast states become instantaneous as $\eps\ra 0$. While at times, $\eps$ can be explicitly decided by the controller design, e.g. in high gain feedback \cite{1977YoungHGF, 1987KhalilHGF}, or as presented in \cite{picallo2021predictive}, the parameter $\eps$ is more commonly used to model the difference in time constants of the two modes \cite{kokotovic}, as will be seen and used in this paper.

Since the dynamics of $z$ are assumed to be much faster than $x$, a boundary layer can be defined based on the quasi-steady state values of $z$ as a function of $x$.  Assuming the existence of isolated equilibrium point (s) of the fast dynamics:
\begin{align}
    0 = g(x,z),
\end{align}
the quasi-steady state can be defined as:
\begin{align}\label{eq:zeqbr}
    z = h(x),
\end{align}
i.e. $h(x)$ satisfies the equation $0 = g\big(x, \;h(x)\big).$ We can then define the deviation of the fast state to this quasi-steady state (known as the boundary layer dynamics) as follows:
\begin{align}\label{eq:qssz}
    y=z-h(x),
\end{align}
which shifts the fast state equilibrium to $\bar{y}=0$. The original system \eqref{eq:sing_model} can be analyzed using the slow states $x$ and $y$ to obtain \cite{kokotovic, hkkhalil}:
\begin{align}\label{eq:boundary_singular}\renewcommand{\arraystretch}{1.15}
\begin{array}{rl}
    \cfrac{dx}{dt}&=f(x,y+h(x)),\\
\epsilon\cfrac{dy}{dt}&=g(x,y+h(x))-\epsilon\dfrac{dh}{dx}f(x,y+h(x)),
\end{array}
\end{align}
where the last term is obtained by differentiating \eqref{eq:qssz} with respect to time.
It is useful to consider a reduced order model by assuming $\eps=0$, i.e. fast states are instantaneous, and considering a different time scale $\tau = t/\eps$, implying $d\tau = dt/\eps$, to obtain:
\begin{align}\renewcommand{\arraystretch}{2}\label{eq:redord}
    \tbf{(reduced system)}\;\; \left\{\begin{array}{l}
         \dfrac{dx}{dt} = f(x,\;h(x)) \\
         \dfrac{dy}{d\tau} = g(\bar{x},\;y+h(\bar{x}))\;\;\
    \end{array}\right.
\end{align}
where $\bar{x}$ is assumed to be constant from the perspective of $\tau$ since $\eps\ra0$.

\subsection{Singular Perturbation for Linear Systems:}
The analysis demonstrated in the previous subsection can be applied to linear systems or a linearized version of \eqref{eq:sing_model}.  In this case, the singular perturbation model is defined as:
\begin{align}\label{eq:sing_lin}\renewcommand{\arraystretch}{2}
\begin{array}{rl}
     \dfrac{dx}{dt} \s= A_{11}x + A_{12}z, \\
     \eps\dfrac{dz}{dt} \s= A_{21}x + A_{22}z.
\end{array}
\end{align}
The quasi-steady state of $z$, described in \eqref{eq:qssz}, can be computed as follows:
\begin{align}\label{eq:h_linear}
    0 = A_{21}x + A_{22}z\Ra z = h(x) = -A_{22}^{-1}A_{21}x. 
\end{align}
The system analyzed in the new coordinates, described by the slow states and the boundary layer dynamics $y = z-h(x)$, is derived as:
\begin{align}
    &\begin{pmatrix}\dot{x}\\\epsilon\dot{y}\end{pmatrix}=&\begin{pmatrix}
    A_{11}-A_{12}A_{22}^{-1}A_{21}&A_{12}\\\Sigma_1&A_{22}+\Sigma_2
    \end{pmatrix}\begin{pmatrix}
    x\\y
    \end{pmatrix},
\end{align}
where
\begin{align}\renewcommand{\arraystretch}{1.25}
\begin{array}{rl}
    \Sigma_1 &= \epsilon \left[A_{22}^{-1}A_{21}A_{11}-A_{22}^{-1}A_{21}A_{12}A_{22}^{-1}A_{21}\right],\\
    \Sigma_2 & = \epsilon\left[A_{22}^{-1}A_{21}A_{12}\right].
    \end{array}
\end{align}

We can then obtain a reduced system by assuming the fast states are instantaneous and using a different time scale $\tau = t/\eps$ as shown below \cite{kokotovic}:
\begin{align}\renewcommand{\arraystretch}{2}\label{eq:redlinear}
    \tbf{(red. linear sys.)}\left\{\begin{array}{l}
         \dfrac{dx}{dt} = \pa{A_{11}-A_{12}A_{22}^{-1}A_{21}}x+A_{12}y \\ 
         \dfrac{dy}{d\tau} = A_{22}y
    \end{array}\right.
\end{align}
since  $\Sigma_{1},\Sigma_2\ra0$ as $\eps\ra 0$.

\section{Sensitivity Conditioning   Control}
\subsection{Sensitivity Conditioning}
Consider a  general nonlinear system  as  follows:
\begin{align}
    \begin{array}{l}
        \dfrac{dx}{dt} = f(x,\;z),\\
        \dfrac{dz}{dt} = g(x,\;z),
    \end{array}
\end{align}
where no assumptions are made on the time scale separation between the two states.  In addition, $x$ and $z$ are assumed to have the same dimensions as stated in the previous section. The sensitivity conditioning term is defined as the change in quasi-steady state \eqref{eq:zeqbr} with respect to time: $\frac{dh(x)}{dt}$. This derivative can be calculated by implicit function theorem \cite{krantz2002implicit}.

\begin{assumption}
The function $g(x,\;z):\mc{X}\times \mc{Z} \ra \mb{R}^{n_z}$ is continuously differentiable for all $x\in\mc{X}\subseteq \mb{R}^{n_x}$ and $z\in\mc{Z}\subseteq\mb{R}^{n_z}$, where $\mc{X}$ and $\mc{Z}$ are open sets, and $\det|\nabla_z g(x,z)|\neq 0$.
\end{assumption}
\begin{theorem}\label{thm:ift}
{(Implicit Function Theorem \cite{krantz2002implicit})}  Under assumption 1, assume there exists a point $(\bar{x},\;\bar{z})$ such that $g(\bar{x},\;\bar{z}) = 0$. Then, there exists neighborhoods $\tilde{\mc{X}}$ and $\tilde{\mc{Z}}$ containing $\bar{x}$ and $\bar{z}$ respectively, such that there exists a function $h:\tilde{\mc{X}}\ra\tilde{\mc{Z}}$ satisfying:
\begin{align}
    g(x,\;h(x)) = 0\;\;\;\forall x\in\tilde{\mc{X}}
\end{align}
\end{theorem}

The implicit function theorem thus implies the existence of the function $z=h(x)$ as described in \eqref{eq:zeqbr}.  In addition, Theorem \ref{thm:ift} can be used to obtain the following derivative:
\begin{align}\label{eq:hderiv}
    \frac{dh(x)}{dx} = -\pa{\nabla_z g(x,\;z)}^{-1}\nabla_x g(x,\;z) \teq S^{z}_x(x,\;z),
\end{align}
where $S^{z}_x(x,\;z)$ is defined as the sensitivity of quasi steady state, $z=h(x)$, as a function of $x$ in \cite{picallo2021predictive}. The main advantage of \eqref{eq:hderiv} is that the function $h(x)$ is not explicitly needed to obtain its derivative with respect to $x$. The derivative of $h(x(t))$ with respect to time is obtained using \eqref{eq:hderiv} as follows:
\begin{align}\label{eq:senscond}
    \frac{dh(x(t))}{dt} = \frac{dh}{dx}\frac{dx}{dt} = S^z_s(x,\;z)\frac{dx}{dt}= S^z_s(x,\;z)f(x,\;z).
\end{align}
The  sensitivity  conditioning   term proposed in \cite{picallo2021predictive} can now be obtained as follows:
\begin{align}\label{eq:pred_model}\renewcommand{\arraystretch}{1.15}
\tbf{(Sensitivity Cond.)}\left\{\begin{array}{rl}
    \dfrac{dx}{dt}&=f(x,z),\\
    \dfrac{dz}{dt}&=g(x,z)+\underbrace{S_{x}^z(x,z)}_{\frac{dh}{dx}}\underbrace{f(x,z)}_{\frac{dx}{dt}}.
\end{array}\right.
\end{align}
The additional term, $\frac{dh}{dt}$, added to the fast states dynamics can be seen as a   term anticipating the dynamics of the quasi-steady state $h(x(t))$. 

Comparing \eqref{eq:pred_model} with \eqref{eq:boundary_singular}, the additional sensitivity term attempts to cancel the same component in the boundary layer dynamics. This can be seen by applying the change of coordinates:
\begin{align}
    y=z-h(x),
\end{align}
to obtain the boundary layer dynamics for the model with  sensitivity conditioning \eqref{eq:pred_model}:
\begin{align}\renewcommand{\arraystretch}{1.15}
\s\left\{\begin{array}{rl}
    \dfrac{dx}{dt}&=f(x,y+h(x)) \\
    \dfrac{dy}{dt}&=g(x,y+h(x))+{\dfrac{dh}{dx}}f(x,y+(x)) \\&\q-\cfrac{dh}{dx}f(x,y+h(x)),
\end{array}\right. \\
\Ra
\s \left\{\begin{array}{rl}
     \dfrac{dx}{dt}&=f(x,y+h(x))\\
    \dfrac{dy}{dt}&=g(x,y+h(x)),
\end{array}\right.
\end{align}
\noindent which is similar to \eqref{eq:redord}  for the case when $\epsilon \rightarrow 0$, i.e. the reduced order system, without the need to introduce a different time scale.

\subsection{Analysis for Systems with  Time Scale Separation}\label{sec:singpert}
We now consider the addition of sensitivity conditioning term \eqref{eq:hderiv} to systems that exhibit large degree of time scale separation i.e. with fast and slow dynamics, modeled in singular perturbation form as shown in \eqref{eq:sing_model}. The singular perturbation parameter, $\eps$, is used to model different time scale separation between the fast and slow states. It will be demonstrated in the simulation results that a smaller $\eps$ is associated with a larger separation between the eigenvalues of fast and slow states.

The system in \eqref{eq:sing_model} is modified as follows to compute the sensitivity term:
\begin{align}\renewcommand{\arraystretch}{1.15}\label{eq:newsingsys}    
    \begin{array}{rl}
        \dfrac{dx}{dt}&=f(x,\;z)\\
        \dfrac{dz}{dt}&=\dfrac{1}{\epsilon}g(x,\;z) = \tilde{g}(x,\;z). 
   \end{array}
\end{align}
The sensitivity term \eqref{eq:hderiv} is computed for \eqref{eq:newsingsys} as follows:
\begin{align}\renewcommand{\arraystretch}{1.15}
S^{z}_x(x,\;z) &= -\pa{\nabla_z \tilde{g}(x,\;z)}^{-1}\nabla_x \tilde{g}(x,\;z) \\
\Ra S^z_{x}(x,\;z) \s= -\eps\pa{\nabla_z g(x,\;z)}^{-1}\frac{1}{\eps}\nabla_x g(x,\;z) \\
\Ra S^{z}_x(x,\;z)\s=-\pa{\nabla_z g(x,\;z)}^{-1}\nabla_x g(x,\;z).
\end{align}
Therefore, it can be seen that the singular perturbation parameter, $\eps$, does not change the sensitivity term for systems modeled in singular perturbation form. Adding the sensitivity conditioning \eqref{eq:senscond} to \eqref{eq:newsingsys}, we can obtain:
\begin{align}
    \s\left\{\begin{array}{rl}
        \dfrac{dx}{dt}&=f(x,z),\\
        \dfrac{dz}{dt}&=\dfrac{1}{\eps}g(x,z)+\underbrace{S_{x}^z(x,z){f(x,z)}}_{\frac{dh}{dt}},
    \end{array}\right.\\
    \Ra\s\left\{\begin{array}{rl}
        \dfrac{dx}{dt}&=f(x,z),\\
        \eps\dfrac{dz}{dt}&=g(x,z)+\eps \underbrace{S_{x}^z(x,z){f(x,z)}}_{\frac{dh}{dt}}.
    \end{array}\right. \label{eq:singsens}
\end{align}

We can now state the following proposition regarding the impact of sensitivity conditioning term for systems with  time scale separation.
\begin{proposition}\label{prop:epssens}
     The impact of sensitivity conditioning term, $\frac{dh}{dt} = S^{z}_x(x,\;z)f(x,\;z)$, when added to the fast states dynamics, $\dot{z}$, converges to zero as $\eps\ra 0$ i.e. as the time scale separation between slow ($x$) and fast states ($z$) grows, the impact of the sensitivity conditioning is reduced.
\end{proposition}
\begin{proof}
It can be seen from \eqref{eq:singsens} that the fast dynamics are simplified as follows
\begin{align}
    \epsilon\frac{dz}{dt}=g(x,z)+\epsilon S_x^z(x,z)f(x,z),
\end{align}
where the right hand term, $\epsilon S_x^z(x,z)f(x,z)$, converges to zero as $\eps\ra 0$ i.e. the impact of this term to the system dynamics is reduced as the time scale separation grows.
\end{proof}

Hence, it can be concluded that the addition of sensitivity conditioning term improves the performance of systems as long as the time scale separation between slow ($x$) and fast ($z$) states is not significant.

\subsection{Sensitivity Conditioning for Linear Systems}
Without making any assumptions on the time scale separation between the two states $x$ and $z$, we define a linear system as follows:
\begin{align}\label{eq:lin}\renewcommand{\arraystretch}{1.4}
    \begin{pmatrix}\dot{x}\\\dot{z}\end{pmatrix}=\begin{pmatrix}
    f(x,z)\\g(x,z)
    \end{pmatrix}=\begin{pmatrix}
    A_{11}&A_{12}\\A_{21}&A_{22}
    \end{pmatrix}\begin{pmatrix}
    x\\z
    \end{pmatrix}.
\end{align}
The sensitivity of the quasi-steady state of $z$ with respect to $x$ can be computed using \eqref{eq:sing_lin}, \eqref{eq:hderiv} as follows:
\begin{align}
    S_{x}^z(x,z)=-A_{22}^{-1}A_{21}.
\end{align}
The sensitivity conditioning term, $\frac{dh}{dt}$, can be simplified as:
\begin{align}\label{eq:lin_ps_term}
S_{x}^z(x,z)f(x,z)=-A_{22}^{-1}A_{21}\left[A_{11}x+A_{12}z\right].
\end{align}
Adding this   term to the linear system in \eqref{eq:lin}, similar to \eqref{eq:pred_model}, the dynamics are modified as shown:
\begin{align}\label{eq:linearwsc}
\pmt{\dot{x}\\ \dot{z}} = \pmt{A_{11}\s A_{12} \\ 
A_{21}-A_{22}^{-1}A_{21}A_{11}\s A_{22}-A_{22}^{-1}A_{21}A_{12}}\pmt{x\\ z}.
\end{align}
The advantages of adding a sensitivity conditioning term are not directly seen from \eqref{eq:linearwsc}. A deeper insight is  gained by changing coordinates to the boundary layer dynamics, $y = z-h(x)$. This results in the following proposition (similar to proposition 2 in \cite{picallo2021predictive}).

\begin{proposition}\label{prop:prop2}
     The linear system with sensitivity conditioning in \eqref{eq:linearwsc} has the following eigenvalues:
     \begin{align}
         \sigma(A_{11}-A_{12}A_{22}^{-1}A_{21})\cup \sigma(A_{22})
     \end{align}
     where $\sigma(A)$ is defined as the  set of eigenvalues of a square matrix $A$.
\end{proposition}
\begin{proof}
The result can be obtained by performing the following change of coordinates based on the boundary layer dynamics, $y = z-h(x) = z+A_{22}^{-1}A_{21}x$, as follows:
\begin{align}\label{eq:transformation}
    \pmt{x\\ z} = \pmt{I \s 0 \\ -A_{22}^{-1}A_{21}\s I}\pmt{x \\ y}\teq T \pmt{x \\ y}
\end{align}
to the system  in \eqref{eq:linearwsc}. In the new coordinates, we can obtain:
\begingroup
\begin{align}\label{eq:coordchange}
    \pmt{\dot{x} \\ \dot{y}} = T^{-1} A_{sc}T\pmt{x \\ y}
\end{align}
\endgroup
where
\begin{align}
    A_{sc} \teq \pmt{A_{11}\s A_{12} \\ 
A_{21}-A_{22}^{-1}A_{21}A_{11}\s A_{22}-A_{22}^{-1}A_{21}A_{12}}.
\end{align}
Simplifying \eqref{eq:coordchange}, we can obtain:
\begin{align}\label{eq:exactxy}
     \pmt{\dot{x} \\ \dot{y}} = \pmt{A_{11}-A_{12}A_{22}^{-1}A_{21} \s A_{12}\\ 0 \s A_{22}}\pmt{x \\ y}
\end{align}
and the desired result follows: $\sigma(T^{-1}A_{sc}T) = \sigma(A_{sc}) = \sigma(A_{11}-A_{12}A_{22}^{-1}A_{21})\cup \sigma(A_{22})$
\end{proof}
Therefore, it can be seen that the additional sensitivity conditioning term helps the closed loop system achieve the same behavior as the reduced order model shown in \eqref{eq:redlinear}.

\section{  Implementation for Linear Systems}
The implementation of sensitivity conditioning term for linear systems is discussed in this section. \blue{The analysis presented can be applied directly to linear systems. While it can also be used in linearized models of nonlinear systems, the sensitivity conditioning term may have to be recomputed at every linearization point (i.e. steady state).} Furthermore, the addition of this   term is shown to be dependent on the number of inputs in the system and in most cases an exact implementation is not feasible.  For this case,  an approximate implementation is proposed.

\subsection{Exact Sensitivity Conditioning}
Prior to the implementation of the sensitivity conditioning term, an open loop linear system is defined as follows:
\begin{align}\label{eq:linsyswctrl}
\begin{array}{rl}
    \dot{x}&=A_{11}x+A_{12}z,\\
    \dot{z}&=\tilde{A}_{21}x+\tilde{A}_{22}z+Bu,
    \end{array}
\end{align}
where $x\in\mb{R}^{n_x},\;z\in\mb{R}^{n_z}$, $u\in\mb{R}^m$, $B\in\mb{R}^{n_z\times m}$, and the matrices $A_{11},\;A_{12},\;\tilde{A}_{21}$, $\tilde{A}_{22}$ are of appropriate dimensions.  

\begin{figure}[!t]
	\begin{center}
		\centering
		\includegraphics[width=0.47\textwidth]{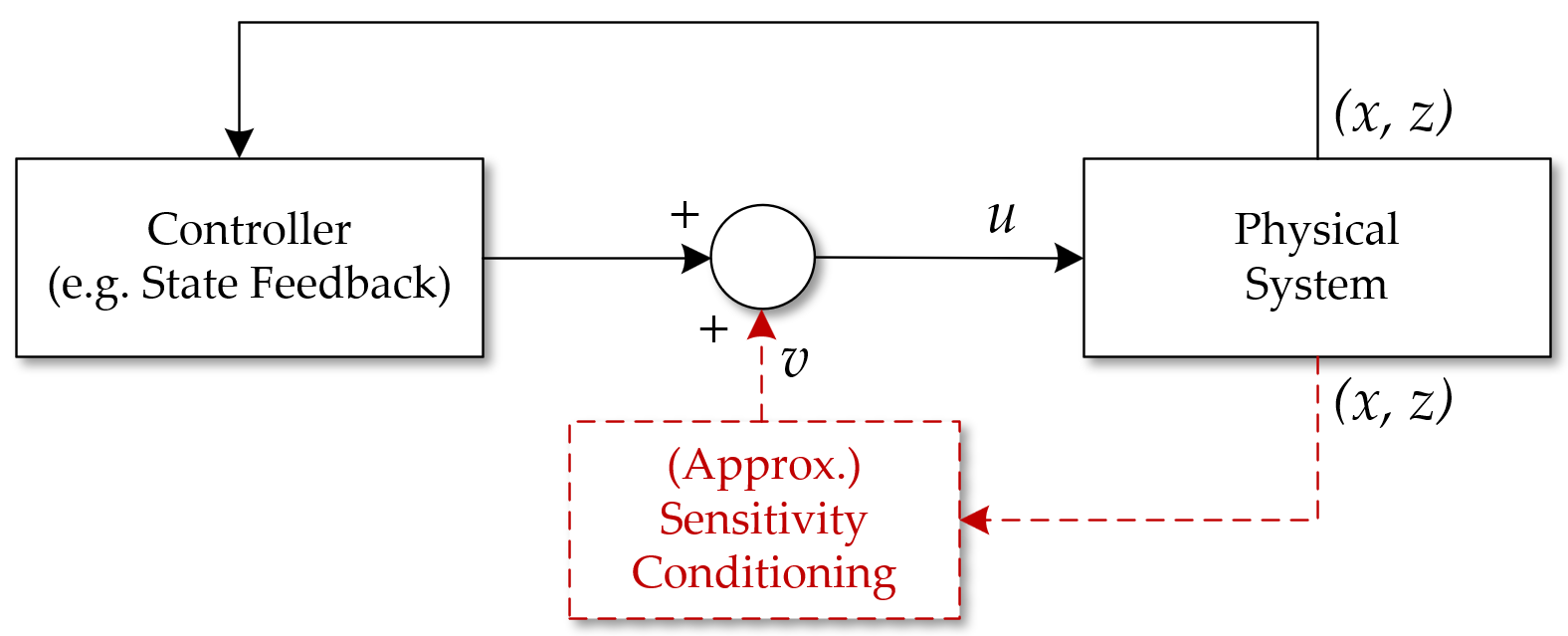}
		\caption{High level overview of sensitivity conditioning   implementation.}\label{fig:overviewimpl}
	\end{center}
\end{figure}

Assume the input is a state feedback controller of the form:
\begin{align}\label{eq:linctrl}
    u = K_x x+ K_zz + v
\end{align}
where $K_x$ and $K_z$ are chosen such that the closed loop system is asymptotically stable and the new input $v\in\mb{R}^m$ will be used to implement the   sensitivity conditioning term, as shown in Fig. \ref{fig:overviewimpl}.

Adding \eqref{eq:linctrl} to \eqref{eq:linsyswctrl}, we obtain the following closed loop system:
\begin{align}\label{eq:closed_ps}
\begin{array}{rl}
\dot{x}&=A_{11}x+A_{12}z,\\
\dot{z}&=A_{21}x+A_{22}z+Bv,
\end{array}
\end{align}
where
\begin{align}
\begin{array}{rl}
A_{21}&= \tilde{A}_{21}+BK_x,\\    
A_{22}&= \tilde{A}_{22}+BK_z.
\end{array}
\end{align}

The analysis in section 2 (singular perturbation) corresponds to \eqref{eq:closed_ps} when $v=0$.  The addition of the sensitivity conditioning term shown in \eqref{eq:hderiv} and \eqref{eq:pred_model}, will be decided by the new input $v\in\mb{R}^m$.
\begin{assumption}
The $\rank{B} = \dim(u) = m$ whenever $m\leq n_z$ and the $\rank{B} = n_z$ for $m > n_z$.
\end{assumption}
Therefore,  the   term, $v$, can be obtained as follows:
\begin{align}\label{eq:PS_term}
    Bv(t)&=S_x^z\big(x(t),\;z(t)\big)f\big(x(t),\;z(t)\big) && \txt{(general form)},\\
    Bv(t)&=-A_{22}^{-1}A_{21}\left[A_{11}x(t)+A_{12}z(t)\right]&&\txt{(linear case)}.\label{eq:pslinear}
\end{align}

Solving for $v$ thus amounts to computing the solution of a linear system of equations for different values of $x(t)$ and $z(t)$.  
The system of equations shown in \eqref{eq:PS_term} or \eqref{eq:pslinear} has a solution for $v(t),\;\forall t\in\mb{R}$, not necessarily unique, if the linear mapping defined by the matrix $B$, i.e. $L(v)\teq Bv:\mb{R}^m\ra\mb{R}^{n_z}$, is onto.
By definition, the mapping $L(v) = Bv$ is onto if for all $b\in\mb{R}^{n_z}$ there exists a $v\in\mb{R}^m$ such that $Bv = b$, where in this case, $b = S(x,\;z)f(x,\;z)$.
We can now consider two cases, associated with the dimension of $B$ and $u$:
\begin{itemize}
    \item \tbf{Case 1}: The number of inputs $u$ is equal to the dimension of $z$, i.e. $m=n_z$. In this case, the matrix $B\in\mb{R}^{n_z\times n_z}$  is invertible by assumption 2.  The unique solution to \eqref{eq:pslinear} is then:
    \begin{align}\label{eq:case1}
        v = B^{-1}S_x^z(x,z)f(x,z).
    \end{align}
    \item \tbf{Case 2}:  Now consider the case where $m>n_z$, i.e. there are more inputs $u$ than  states $z$. Under assumption 2, this corresponds to an under-determined system of equations with an infinite number of solutions. These can be characterized using the Moore-Penrose pseudoinverse as follows \cite{barata2012moore}:
    \begin{align}\label{eq:case2}
        v = B^{\dagger}_RS_x^z(x,z)f(x,z) + (I-B^\dagger_R B)p
    \end{align}
    where $B^\dagger_R \teq B^T\pa{BB^T}^{-1}$ and $p\in\mb{R}^m$ is an arbitrary vector (can be set to zero). The subscript $R$ is used to denote that $B$ is right invertible i.e. $BB^\dagger_R =  I$.
\end{itemize}
For the  two cases shown in \eqref{eq:case1} and \eqref{eq:case2}, the term $v$ satisfies \eqref{eq:PS_term} or \eqref{eq:pslinear} exactly.  This implies that eigenvalues of the closed loop system with the   conditioning term are the same as those shown in Proposition \ref{prop:prop2}. In addition, an exact implementation of the sensitivity conditioning term can be achieved only if the number of inputs is greater than or equal to the number of states $z$.  

\subsection{Approximate Sensitivity Conditioning}
For most practical cases, the number of inputs/unknowns given by $\dim(u) = m$, is less than $\dim(z)=n_z$ or number of equations, i.e. $m<n_z$. This case corresponds to an over-determined system of equations where generally no solution exists. Therefore, assuming the range of sensitivity conditioning term \eqref{eq:hderiv} spans all of $\mb{R}^{n_z}$, it is not possible to find a   vector, $v\in\mb{R}^m$, satisfying \eqref{eq:PS_term} or \eqref{eq:pslinear}.

In this case, we consider a least squares approximation to the system of equations in \eqref{eq:PS_term} or \eqref{eq:pslinear} as follows:
\begin{align}\label{eq:lsprob}
    \min_{v\in\mb{R}^m}\;\;||Bv-S^z_x(x,\;z)f(x,\;z)||_2^2
\end{align}
which has an analytical solution of the form:
\begin{align}\label{eq:apprxsc}
    v = B^{\dagger}_LS_x^z(x,z)f(x,z),
\end{align}
where $B^\dagger_L = \left(B^TB\right)^{-1}B^T,$ and the subscript $L$ is used to denote $B$ is left invertible i.e. $B^\dagger_LB = I$. However,  it is emphasized the approximate   term $v$ in \eqref{eq:apprxsc} satisfies:
\begin{align}
    Bv = S^z_x(x,\;z)f(x,\;z) + e(x,\;z),
\end{align}
where $e(x,\;z)\in\mb{R}^{n_z}$ is an error term whose two-norm squared is minimized by \eqref{eq:lsprob}. For linear systems, using the sensitivity conditioning term shown in \eqref{eq:pslinear}, the error term can be simplified as follows:
\begin{align}\label{eq:errorgeneral}
    e(x,\;z) \s= Bv - S^z_x(x,\;z)f(x,\;z)\\
    \Ra e(x,\;z) \s = (I-BB^\dagger_L)A_{22}^{-1}A_{21}\underbrace{\pa{A_{11}x+A_{12}z}}_{\dot{x}(t)}.\label{eq:errorlinear}
\end{align}

We can analyze the closed loop system with the approximate sensitivity conditioning term \eqref{eq:apprxsc}.  Substituting \eqref{eq:apprxsc} in \eqref{eq:closed_ps}, we obtain the following:
\begin{align}
\s\left\{
\begin{array}{rl}
\dot{x}&=A_{11}x+A_{12}z\\
\dot{z}&=A_{21}x+A_{22}z+BB^{\dagger}_LS_x^z(x,z)f(x,z).
\end{array}\right. \\\label{eq:linearwasc}
\Ra \s\left\{\begin{array}{rl}
\dot{x}&=A_{11}x+A_{12}z \\
\dot{z}&=\pa{A_{21}-BB^{\dagger}_LA_{22}^{-1}A_{21}A_{11}}x\\&+\pa{A_{22}-BB^{\dagger}_LA_{22}^{-1}A_{21}A_{12}}z.
\end{array}\right.
\end{align}
 
\begin{proposition}
     The steady state of closed loop linear system in \eqref{eq:lin} or \eqref{eq:closed_ps} with $v=0$, is the same as the steady state for the system in \eqref{eq:linearwasc} with approximate sensitivity conditioning term. 
\end{proposition}
\begin{proof}
    Define the steady state for \eqref{eq:lin} without approximate sensitivity conditioning as $\pa{x^*,\;z^*}$, satisfying:
    \begin{align}\label{eq:sstate}
        \s\begin{array}{rl}
            0 \s= A_{11}x^* + A_{12}z^*,\\
            0 \s= A_{21}x^* + A_{22}z^*.
        \end{array}
    \end{align}
    With approximate sensitivity conditioning, the steady state can be computed as:
    \begin{align}
        \begin{array}{l}
            0 = A_{11}x^* + A_{12}z^*,\\
            0 = A_{21}x^* + A_{22}z^* - BB^\dagger_LA_{22}^{-1}A_{21}\underbrace{\pa{A_{11}x^* + A_{12}z^*}}_{=0},
        \end{array}
    \end{align}
    which implies that steady state $(x^*,\;z^*)$ satisfies the same equations as \eqref{eq:sstate}.
\end{proof}

\begin{proposition}
     The error term defined in \eqref{eq:errorgeneral} and \eqref{eq:errorlinear} is zero at steady state.
\end{proposition}
\begin{proof}
    This result follows from \eqref{eq:sstate} since at steady state, $(A_{11}x^* + A_{12}z^*) =0$, which implies:
    \begin{align}
        e(x^*,\;z^*) \s = (I-BB^\dagger_L)A_{22}^{-1}A_{21}\underbrace{\pa{A_{11}x^*+A_{12}z^*}}_{=0}\\
         \s = 0.
    \end{align}
\end{proof}

A change of coordinates defined by the boundary layer dynamics, $y = z-h(x) = z+ A_{22}^{-1}A_{21}x$, whose transformation matrix $T$ is shown in \eqref{eq:transformation}, can be applied to simplify dynamics of the closed loop system with approximate sensitivity conditioning  in \eqref{eq:linearwasc}. This system is transformed to:
\begin{align}\label{eq:approxxy}
    \left\{\begin{array}{rl}
    \dot{x} \s=\pa{A_{11} - A_{12}A_{22}^{-1}A_{21}}x + A_{12}y \\
    \dot{y} \s= A_{22}y + e(x,\;y) 
    \end{array}\right.
\end{align}
where $e(x,\;y)$ is obtained from \eqref{eq:errorlinear} as follows:
\begin{align}\label{eq:errorxy}
\begin{array}{rl}
    e(x,\;y) \s= (I-BB^\dagger_L)(A_{22}^{-1}A_{21})\pa{A_{11}-A_{12}A_{22}^{-1}A_{21}}x\\
    \s + (I-BB^\dagger_L)\pa{A_{22}^{-1}A_{21}}A_{12}y
\end{array}
\end{align}

Therefore, it can be seen that the approximate sensitivity conditioning dynamics in \eqref{eq:approxxy} are similar to the exact case \eqref{eq:exactxy}, except for the error term. Moreover, if the error is small, the approximation will provide a similar improvement to the closed loop system as the exact case (section 4a).  

\begin{proposition}
     The norm of the error $||e(x,\;z)||_2= ||e(x,\;y)||_2$ of the approximate sensitivity conditioning term defined in \eqref{eq:errorlinear} and \eqref{eq:errorxy} respectively, can be upper-bounded as follows:
     \begin{align}
     \begin{array}{ll}
         ||e(x,\;z)||_2 \s\leq ||(I-BB^\dagger_L)A_{22}^{-1}A_{21}||_2\; ||\dot{x}(t)||_2\\
         \s= ||(I-BB^\dagger_L)A_{22}^{-1}A_{21}||_2\; ||A_{11}x+A_{12}z||_2
    \end{array}
\end{align}
\end{proposition}
\begin{proof}
    This result is obtained by the  sub-multiplicative property of matrix norm, i.e. $||AB||_2\leq ||A||_2||B||_2$, and realizing the right hand side of \eqref{eq:errorlinear} is equal to $\dot{x}$ in either coordinates.
\end{proof}

In particular, the value of
\begin{align}\label{eq:errorestimate}
    ||(I-BB^\dagger_L)A_{22}^{-1}A_{21}||_2
\end{align} 
provides a good estimate of the approximation error, as it is not dependent on time.

\section{Simulation Results}
In this section, the approximate sensitivity conditioning control strategy is implemented for a voltage regulated buck converter and a PMSM with an active rectifier. The baseline controller is based on a cascaded inner/outer loop PI control.  For the buck converter case, the baseline controller is modified to investigate different time scale separations between the slow and fast dynamics, verifying proposition \ref{prop:epssens}.

\begin{figure}[!b]
	\begin{center}
		\centering
		\includegraphics[width=0.47\textwidth]{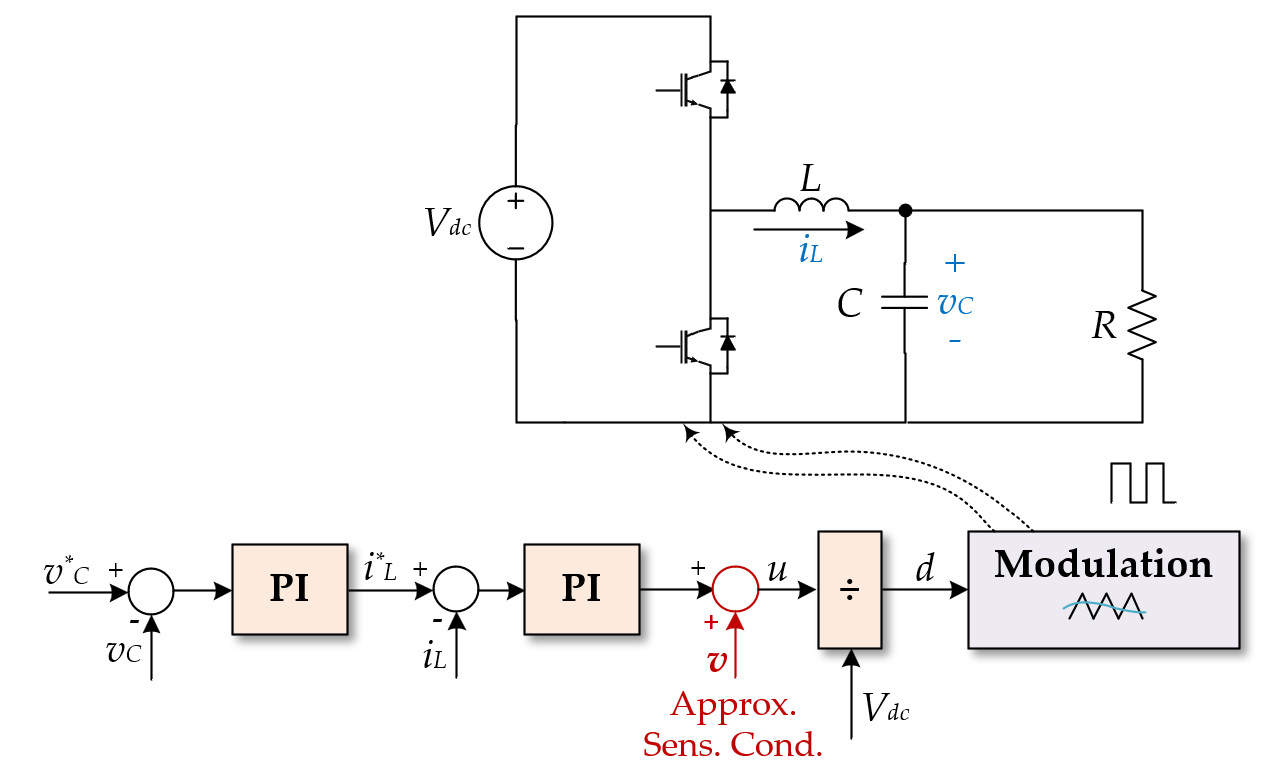}
		\caption{Buck converter with approximate sensitivity conditioning, the term in red is computed using \eqref{eq:buckasc}.}\label{fig:cascade}
	\end{center}
\end{figure}

\begin{table}[!b]
	\renewcommand{\arraystretch}{1.5}
	\caption{DC-DC buck converter parameters used for simulation and experimental results.}
	\label{tab:case1} 
	\centering
	\begin{tabular}{c|c|c}
		\hline \hline
		{\textbf{No}} \s  {\textbf{Parameters}} \s Value\\\hline
		$1$\s Switching frequency \s 40 $\times 10^3$Hz\\
		$2$\s Load resistance (R)\s 18.6 $\Omega$\\
		$3$\s Load capacitance (C)\s 510 $\mu$F\\
		$4$\s Inductance (L)\s 1 mH\\
		$5$\s Outer loop PI gains $(K_{P,v_C},\;K_{I,v_C})$\s 1, 30\\
		$6$\s Inner loop PI gains $(K_{P,i_L},\;K_{I,i_L})$\s 1, 700\\
		\hline \hline
	\end{tabular}
\end{table}
\subsection{Case 1: DC-DC Buck converter}
The open loop state space model of the buck converter shown in Fig. \ref{fig:cascade} is given by 
\begin{align}
    \begin{pmatrix}\dot{v}_C\\\dot{i}_L\end{pmatrix}=\begin{pmatrix}-\frac{1}{RC}&\frac{1}{C}\\-\frac{1}{L}&0\end{pmatrix}\begin{pmatrix}{v}_C\\{i}_L\end{pmatrix}+\begin{pmatrix}0\\\frac{1}{L}\end{pmatrix}{u},
\end{align}
where the input $u\teq d V_\dc$ where $d$ is the duty cycle and $V_\dc$ is the input voltage. Considering a cascaded inner/outer loop PI control, we define two additional states for the integral terms as follows:
\begin{align}
    \dot{\zeta}_{v_C}\s=v_C^r-v_C, \\
    \dot{\zeta}_{i_L}\s=i_L^r-i_L.
\end{align}
The states are now separated into slow $x\in\mb{R}^2$ and fast $z\in\mb{R}^2$ modes:
\begin{align}\label{eq:buck_states}
    x=\begin{pmatrix}v_C\\\zeta_{v_C}\end{pmatrix},~~~~~~~
z=\begin{pmatrix}i_L\\\zeta_{i_L}\end{pmatrix}.
\end{align}
The open loop dynamics of the slow state $x$ can then be written as:
\begin{align}\label{eq:slow_diff}
    \begin{pmatrix}\dot{v}_C\\\dot{\zeta}_{v_C}\end{pmatrix}=\underbrace{\begin{pmatrix}
-\frac{1}{RC}&0\\-1&0
\end{pmatrix}}_{A_{11}}\begin{pmatrix}v_C\\\zeta_{v_C}\end{pmatrix}+&\underbrace{\begin{pmatrix}\frac{1}{C}&0\\0&0\end{pmatrix}}_{A_{12}}\begin{pmatrix}i_L\\\zeta_{i_L}\end{pmatrix}\notag\\+&\begin{pmatrix}0\\1\end{pmatrix}\begin{pmatrix}v_{C}^r\end{pmatrix}.
\end{align}
Similarly, the open loop differential equation for the faster state is given by
\begin{align}\label{eq:buckfastopen}
\begin{pmatrix}\dot{i}_L\\\dot{\zeta}_{i_L}\end{pmatrix}=&\underbrace{\begin{pmatrix}
-\frac{1}{L}&0\\0&0
\end{pmatrix}}_{\tilde{A}_{21}}\begin{pmatrix}v_C\\\zeta_{v_C}\end{pmatrix}+\underbrace{\begin{pmatrix}0&0\\-1&0\end{pmatrix}}_{\tilde{A}_{22}}\begin{pmatrix}i_L\\\zeta_{i_L}\end{pmatrix}\notag\\&+\underbrace{\begin{pmatrix}\frac{1}{L}\\0\end{pmatrix}}_{B}u +\begin{pmatrix}0\\1\end{pmatrix}\begin{pmatrix}i_L^r\end{pmatrix}.
\end{align}
Based on the cascaded PI control shown in Fig. \ref{fig:cascade}, the input $u$ is derived as follows:
\begin{align}
 u=K_{P,i_L}\left(i_L^r-i_L\right)+K_{I,i_L}\left(\zeta_{i_L}\right)+v,
\end{align}
where $v$ will be used for the approximate sensitivity conditioning implementation \eqref{eq:apprxsc}, since the number of inputs (one) is less than the dimension of the fast states (two). The current reference term $i_L^r$ is the output of the outer PI controller of the form: 
\begin{align}
    i_L^r=K_{P,v_{C}}\left(v_{C}^r-v_{C}\right)+K_{I,v_{C}}\left(\zeta_{v_{C}}\right).
\end{align}
In order to compute the sensitivity conditioning term, the closed loop formulation  is obtained by substituting for $u$ and $i_L^r$ as follows:\blue{
\begin{align}\label{eq:buckclosed}
\begin{pmatrix}\dot{i}_L \\ \dot{\zeta}_{i_L}\end{pmatrix} =& \underbrace{\begin{pmatrix} \frac{-1-K_{P,i_L}K_{P,v_C}}{L} &\frac{1}{L}K_{P,i_L}K_{I,v_C} \\ -K_{P,v_C}&K_{I,v_C}\end{pmatrix}}_{A_{21}}\begin{pmatrix}v_C\\\zeta_{v_C}\end{pmatrix} \notag\\
+ &\underbrace{\begin{pmatrix}-\frac{1}{L}K_{P,i_L} & \frac{1}{L}K_{I,i_L} \\ -1 & 0\end{pmatrix}}_{A_{22}} \begin{pmatrix} i_L \\ \zeta_{i_L} \end{pmatrix}\notag\\+&\begin{pmatrix}\frac{K_{P,i_L}K_{P,v_{C}}}{L}\\K_{P,v_{C}}\end{pmatrix}\begin{pmatrix}v_{C}^r\end{pmatrix} + \underbrace{\pmt{\frac{1}{L} \\ 0}}_{B}v.
\end{align}}

\begin{figure}[!b]\newcommand{\meas}{.5}
	\centering
    \includegraphics[width=\meas\textwidth]{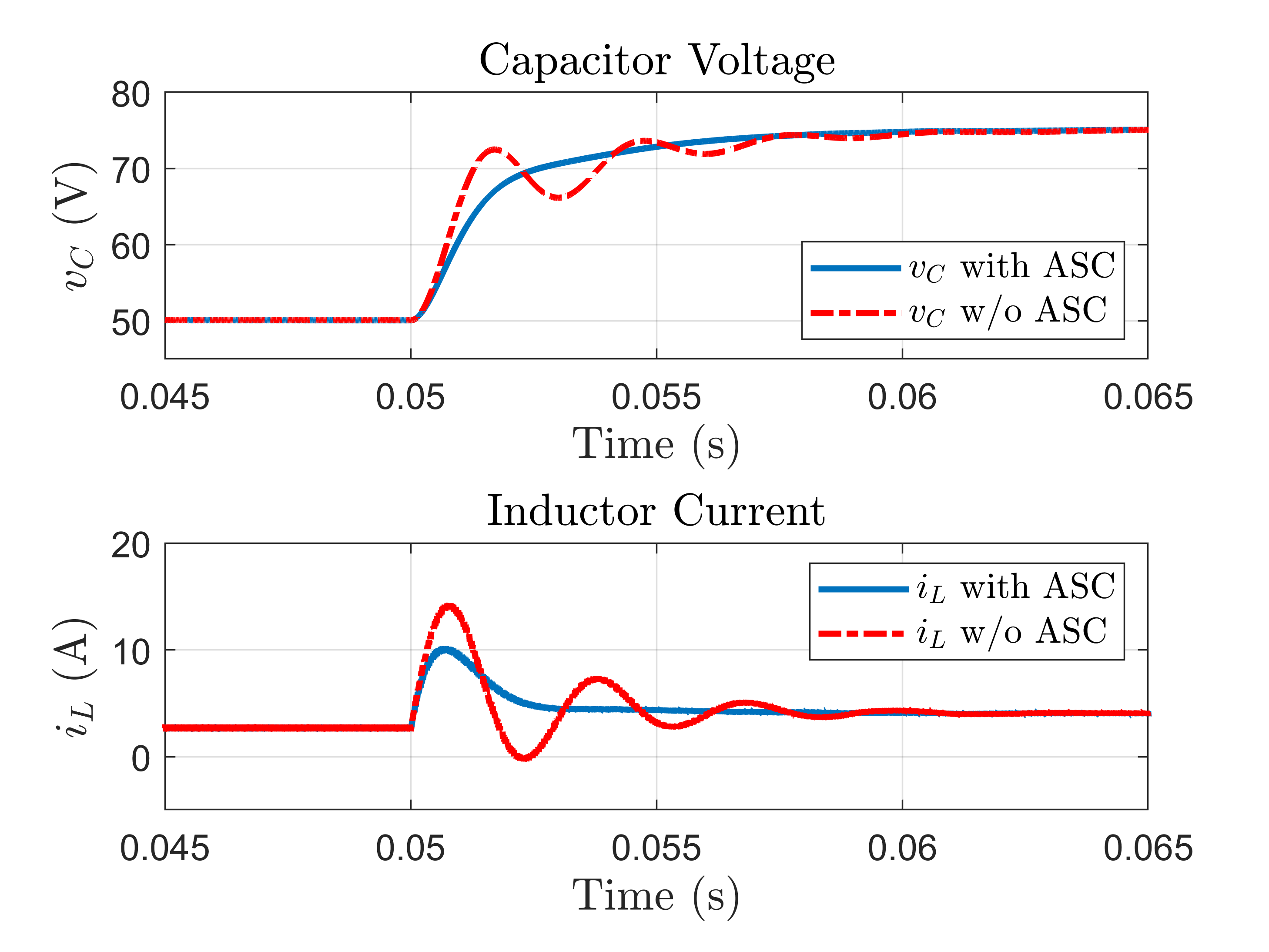}
    \caption{Case 1: Capacitor voltage and inductor current of the buck converter with Approximate Sensitivity Conditioning (ASC) (blue trace) and without (red trace).}\label{fig:sim_case1_zoom}
\end{figure}
The approximate sensitivity conditioning term can then be computed as follows:
\begin{align}
    v \s= B^\dagger_LS^z_x(x,\;z)f(x,\;z)\\
    \Ra v \s= -B^\dagger_LA_{22}^{-1}A_{21}\left[A_{11}x+A_{12}z\right]\label{eq:buckasc}
\end{align}
where $A_{21},\;A_{22}$ are shown in \eqref{eq:buckclosed} and $B$ is defined in \eqref{eq:buckfastopen}.
\begin{figure*}[!h]\newcommand{\meas}{.25}
	\centering
	\centering
	\subfloat[Test 1: Small time scale separation (large $\eps$).]{\label{fig:test1}\includegraphics[height=\meas\textwidth]{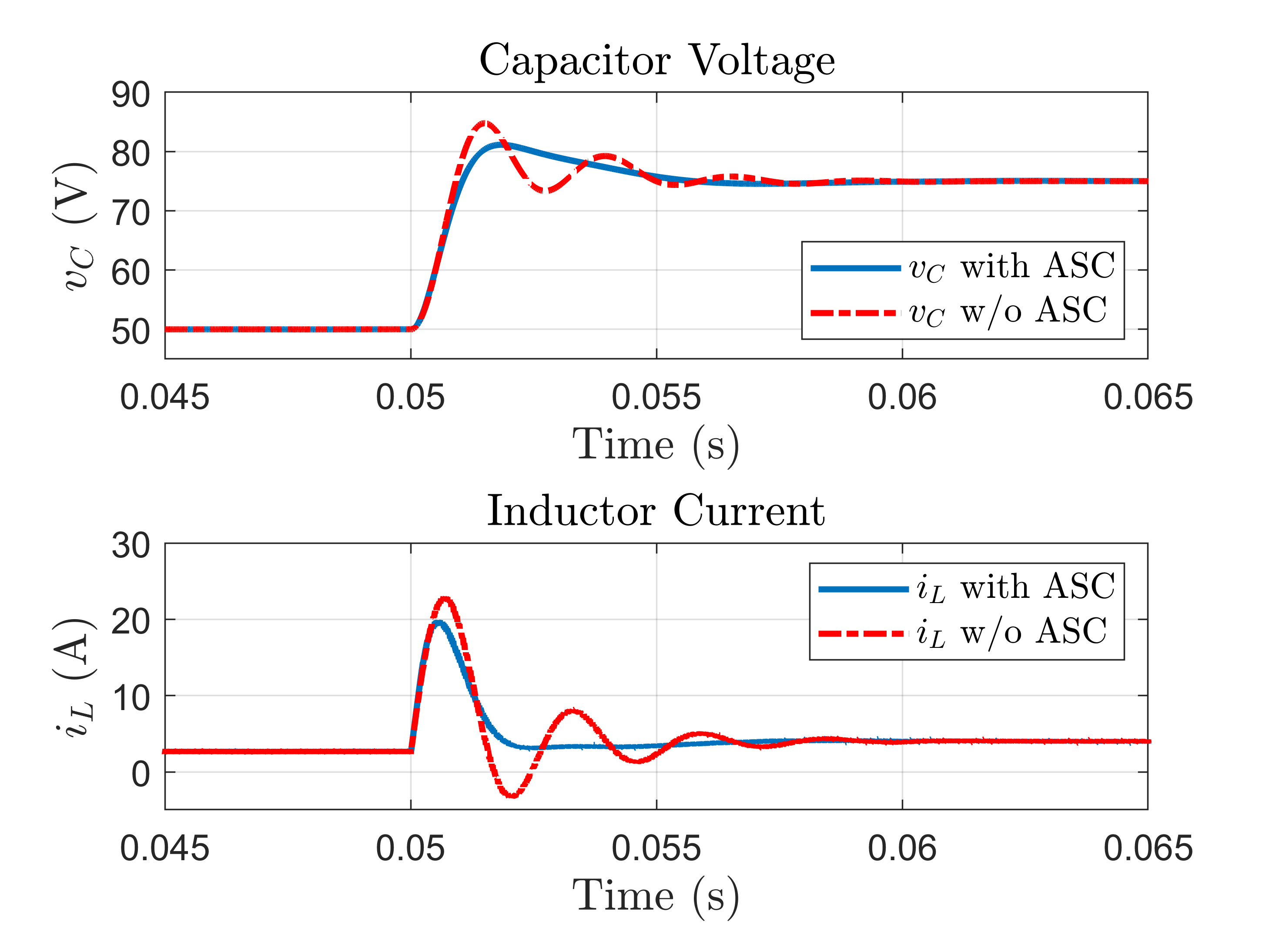}}
	\subfloat[Test 2: Medium time scale separation.]{\label{fig:test2}\includegraphics[height=\meas\textwidth]{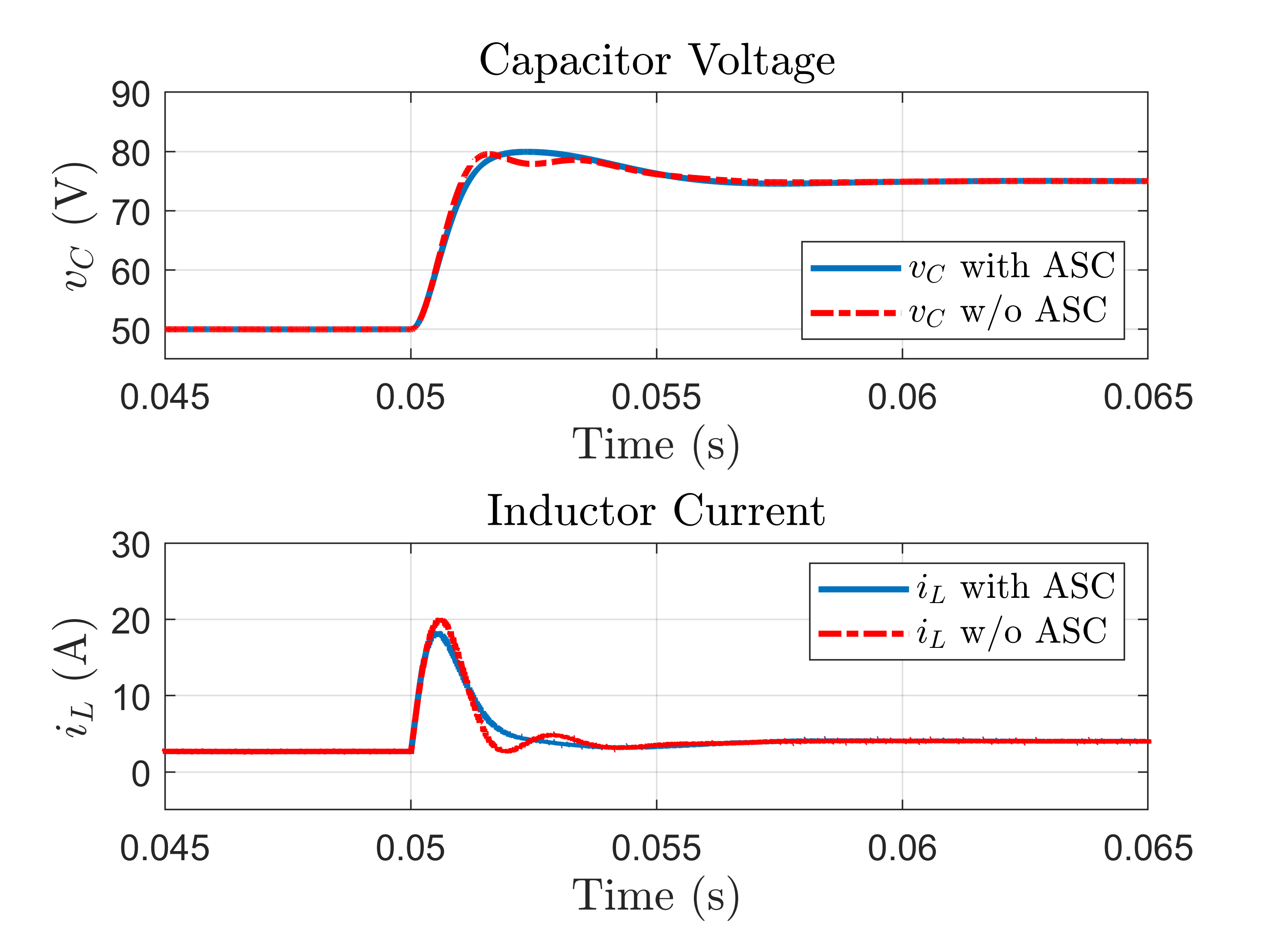}}
	\subfloat[Test 3: Large time scale separation (small $\eps$).]{\label{fig:test3}\includegraphics[height=\meas\textwidth]{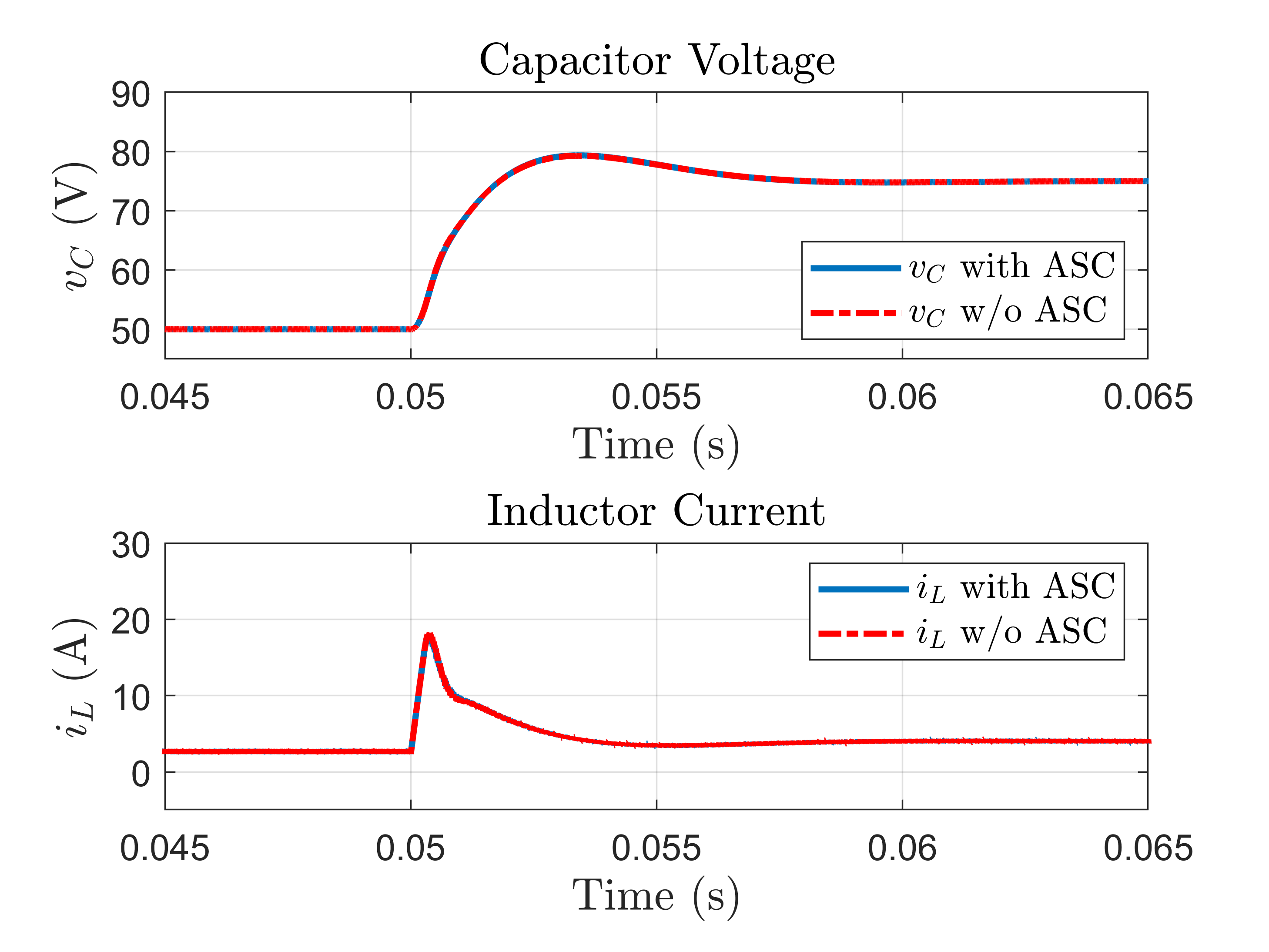}}
	\caption{{Case 2: Simulation results for singular perturbation analysis of buck converter.}}\label{fig:case2}
\end{figure*}
The simulation parameters  are presented in Table \ref{tab:case1}. The load voltage and the inductor current plots are shown in Fig. \ref{fig:sim_case1_zoom}. At time $t = 0.05$ s, the output voltage reference is changed from 50 V to 75 V. During this change in reference value, the  simulation with approximate sensitivity control  provides a better transition to the new reference value in comparison to the regular dual-PI cascade controller.  In this case, the term in \eqref{eq:errorestimate} is $0.0\bar{3}$ which implies that the approximation error is small.

\begin{table}[!t]
	\renewcommand{\arraystretch}{1.25}
	\caption{DC-DC buck converter: PI gains and eigenvalues for case 2.}
	\label{tab:case2} 
	\begin{tabular}{c|c|c|c|c}
		\hline \hline
         \textbf{Test} \s PI (outer) \s PI (inner) \s $\lambda$ w/o sc \s $\lambda$ w/ sc\\\hline \hline
          1 \s (0.94, 970) \s (2, 2000) \s $-474\pm2433j$ \s $-1512\pm2019j$\\
                \s    \s    \s $-579\pm532j$ \s $-463\pm618j$ \\
    \hline
         2 \s (0.7, 574) \s (3, 4500) \s $-1000\pm2670j$ \s $-1755\pm2213j$\\
                \s    \s    \s $-544\pm570j$ \s $-495\pm624j$ \\
    \hline
     3 \s (.45, 255) \s (10, 5e4) \s $-4572\pm5639j$ \s $-5021\pm5211j$\\
                \s    \s    \s $-481\pm493j$ \s $-480\pm498j$ \\
		\hline \hline
	\end{tabular}
\end{table}

\subsection{Case 2: Singular Perturbation Analysis for Buck Converter}
The buck converter presented in case 1 is now implemented with variable PI gains to test the impact of approximate sensitivity control strategy at different time scale separation between the slow ($x$) and fast states ($z$), described in section 3b and proposition \ref{prop:epssens}. The singular perturbation parameter, $\epsilon$, is implicitly associated by eigenvalues of the system \cite{kokotovic}. For example, a smaller $\epsilon$ corresponds to the system having the eigenvalues associated with the fast states ($\lambda_z$) much smaller than the eigenvalues associated with the slow states ($\lambda_x$), i.e.  $\lambda_z \ll \lambda_x$. The PI gains at the inner and outer loop are both used to achieve a different degree of time scale separation. The parameters for this case are shown in Table \ref{tab:case2}, and the simulation results are presented in Fig. \ref{fig:case2}.

In Fig. \ref{fig:test1}, the time scale separation between the slow and fast states is not very large. As seen from Table \ref{tab:case2} test 1,  there is not a clear separation of eigenvalues of the system between slow and fast states without sensitivity conditioning. Therefore, the approximate   term has a higher impact on the dynamics and the new eigenvalues with this term are modified significantly. In  Fig. \ref{fig:test2}, there is a slightly better separation of eigenvalues before the   term is applied (test 2 in Table \ref{tab:case2}).  Therefore, it can be seen that the approximate sensitivity conditioning term has a slightly less impact than the first test.  Lastly, we consider a more drastic separation between the slow and fast states of the system before applying approximate sensitivity conditioning as shown in Fig. \ref{fig:test3} corresponding to test 3 in Table \ref{tab:case2}.  The eigenvalues before and after the   term are relatively similar and the transient response is nearly identical. As mentioned by proposition \ref{prop:epssens}, as the time scale separation between slow and fast eigenvalues increases, the impact of approximate sensitivity   term is reduced. \blue{Furthermore, it can be seen that the (approximate) sensitivity conditioning helps to achieve a similar performance as test 3, even in test 1, without the need of large PI gains, which can saturate the duty cycle.}



\subsection{Case 3: PMSM with Active Rectifier}
In this case, the approximate sensitivity strategy is implemented on a  PMSM with an active rectifier, as shown in Fig. \ref{fig:PMSMoverall} \cite{herrera2021nonlinear, 2016GaoPMSM}.  The system in this case is described by nonlinear differential equations.  However, the equations can be linearized with respect to the input and the analysis presented in section 4b is used, in particular \eqref{eq:apprxsc} to obtain $v$.

The differential equation representing the slower states $x = \br{{v}_{dc},\;{\zeta}_{v_{dc}}}^T$ are:
\begin{align}\renewcommand{\arraystretch}{1.15}\label{eq:PMSMslow}
\begin{array}{rl}
\dot{v}_{dc}&= -\cfrac{1}{RC}v_{dc}+\cfrac{3}{4C}\left(d_d i_d+d_q i_q \right)-\cfrac{1}{C}i_L\\
\dot{\zeta}_{v_{dc}}&= v_{dc}^r-v_{dc},
\end{array}
\end{align}
where $v_{dc}^r$ is the reference voltage and ${\zeta}_{v_{dc}}$ is associated with the outer loop integrator. 

The faster states are defined as $z=\br{i_d,\;\zeta_{id},\;i_q,\;\zeta_{iq}}^T$, with the following nonlinear dynamics:
\begin{align}\renewcommand{\arraystretch}{1.15}\label{eq:PMSMfast}
\begin{array}{rl}
\dot{i}_{d}&=-\cfrac{R_s}{L_d}i_d+\omega_r\cfrac{L_q}{L_d}i_q+\cfrac{1}{2L_d}d_dv_{dc}\\
\dot{\zeta}_{i_d}&= i_d^r-i_d\\
\dot{i}_{q}&= -\cfrac{R_s}{L_q}i_q-\omega_r\cfrac{L_d}{L_q}i_d-\cfrac{\omega_r}{L_q}\lambda_m+\cfrac{1}{2L_q}d_qv_{dc}\\
\dot{\zeta}_{i_q}&= i_q^r-i_q,
\end{array}
\end{align}
where $i_d^r$ and $i_q^r$ are the $d$ and $q$ axis current references respectively (obtained from the outer loop),  $\dot{\zeta}_{i_d}$ and $\dot{\zeta}_{i_q}$ are the $d$ and $q$ axis current error associated with the inner loop PIs, and $u = \br{d_d,\;d_q}^T$.

\begin{table}[!b]
	\renewcommand{\arraystretch}{1.5}\vspace{0pt}
	\caption{PMSM with active rectification: Simulation parameters used in case 3.}
	\label{tab:case3} 
	\centering
	\begin{tabular}{c|c}
		\hline \hline
		  {\textbf{Parameters}} \s Value\\\hline
		Time-step ($T_s$)\s 1 $\mu$s\\
		 $n$ \s 8000 rpm\\
		P, I gains for $v_{dc}$\s 2, 1000\\
		P, I gains for $(i_d)$ and $(i_q)$ \s (0.5, 2) and (0.5, 2)\\
		Poles \s 12 \\
		$L_d$ \s 0.09 mH\\
		 $L_q$  \s  0.255 mH\\
		$r_s$ \s 5.3 m$\Omega$\\
		 $\lambda_m$ \s 0.0385\\
		$\omega_m$ \s 2$\pi\cfrac{n}{60}$\\
		 $\omega_r$ \s $\cfrac{\text{Poles}}{2}\:\omega_m$\\
		Load damp Resistance (R)\s 5e1 $\Omega$ \\	
		Load Capacitance (C) \s 1 mF \\	
		\hline \hline
	\end{tabular}
\end{table}

\blue{The inputs $d_d$ and $d_q$ are obtained from standard inner current loop PIs with decoupling \cite{2016GaoPMSM, herrera2021nonlinear}:
\begin{align}\renewcommand{\arraystretch}{1.5}\label{eq:dddq}
\begin{array}{lll}
d_d&=\cfrac{2}{v_{dc}}\left[K_{P,i_d}\left(i_d^r-i_d\right)+K_{I,i_d}\left(\zeta_{i_d}\right)-\omega_rL_qi_q+v_d\right]\\
d_q&=\cfrac{2}{v_{dc}}\big[K_{P,i_q}\left(i_q^r-i_q\right)+K_{I,i_q}\left(\zeta_{i_q}\right)\\&\q\q\q+\omega_rL_di_d+\omega_r\lambda_m+v_q\big],
\end{array}
\end{align}
where $v = \br{v_d,\;v_q}^T$ will be used for the approximate sensitivity conditioning   term \eqref{eq:apprxsc}.
\begin{figure}[!t]\newcommand{\meas}{.45}
	\centering
	\subfloat[PMSM with an active rectifier circuit and a constant current load.]{\label{fig:pmsm_circuit}\includegraphics[width=\meas\textwidth]{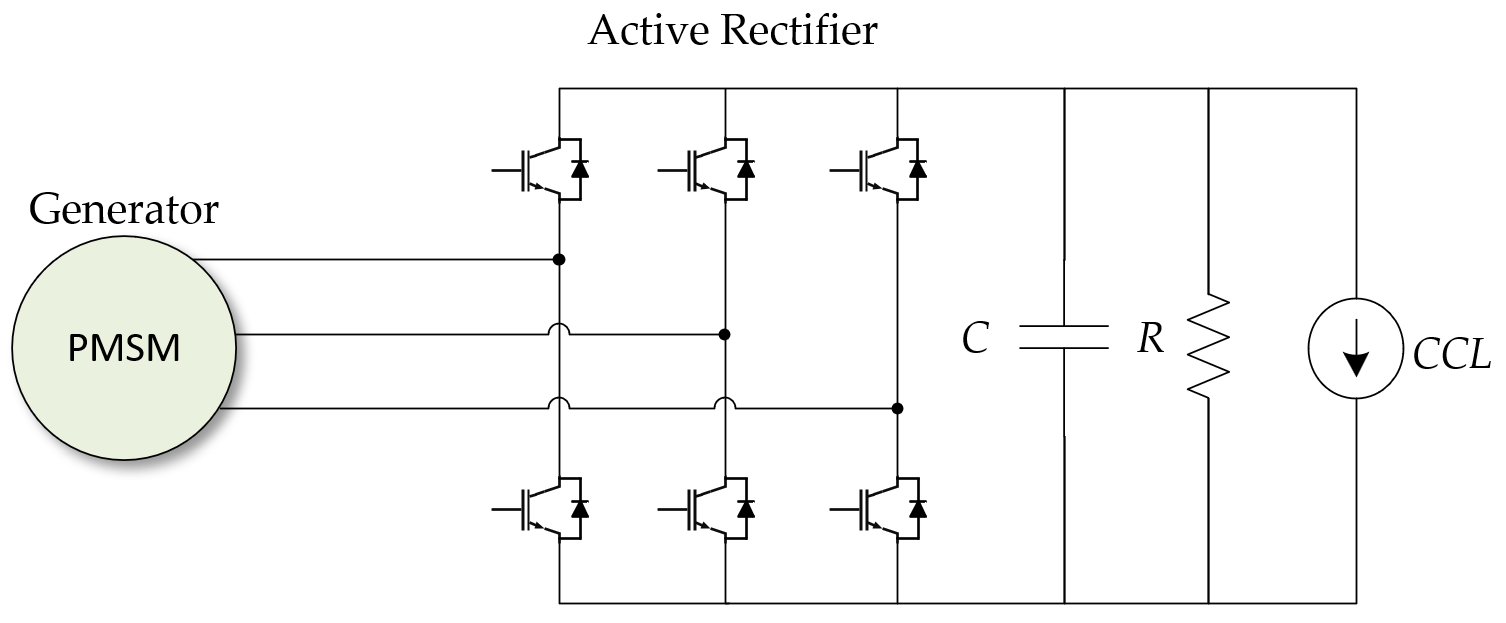}}\\
	\subfloat[Controller overview of PMSM.  The   terms are shown in red.]{\label{fig:pmsm_controllers}\includegraphics[width=\meas\textwidth]{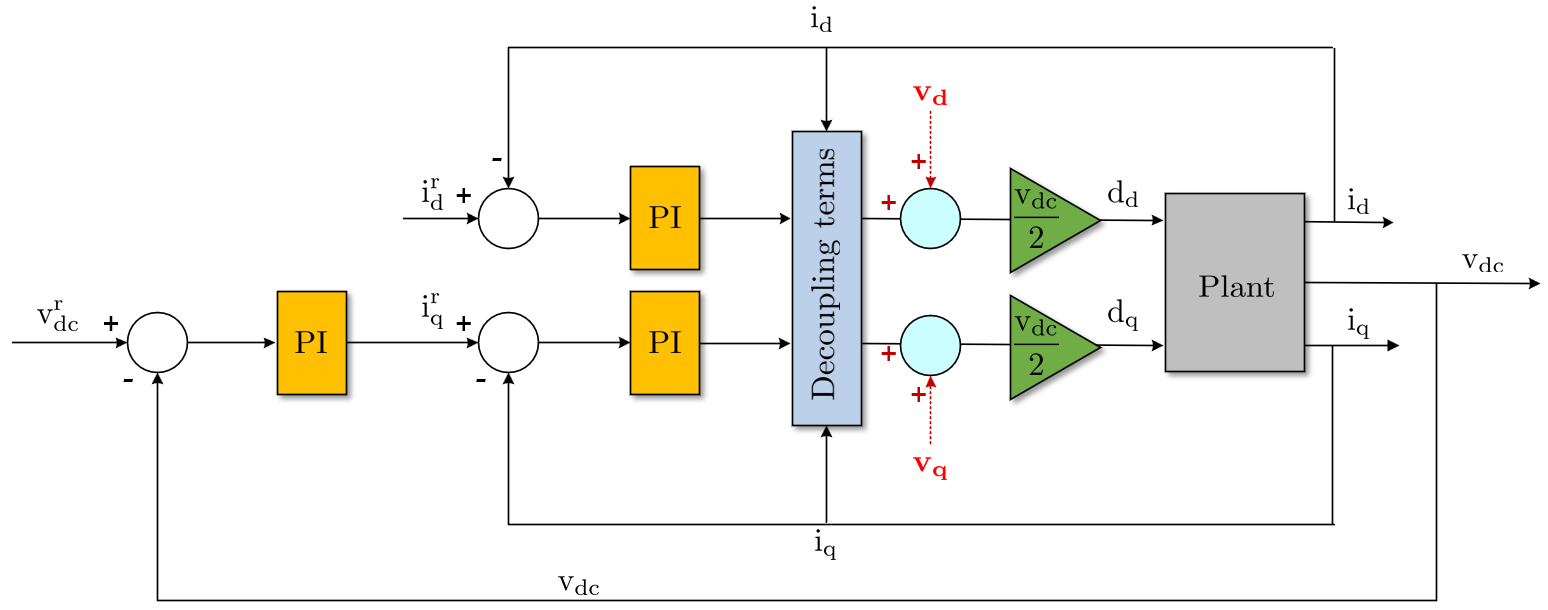}}
	\caption{Approximate sensitivity conditioning strategy for PMSM with active rectification.}\label{fig:PMSMoverall}
\end{figure}
The references for the faster states, $i_d^r$ and $i_q^r$, are obtained from the outer loop PI control as follows:
\begin{align}\renewcommand{\arraystretch}{1.75}\label{eq:iqr}
\begin{array}{rl}
i_d^r&=0\\
i_q^{r}&=-\left[K_{P,v_{dc}}\left(v_{dc}^r-v_{dc}\right)+K_{I,v_{dc}}\left(\zeta_{v_{dc}}\right)\right].
\end{array}
\end{align}
where $i_d^r=0$ corresponds to the case without flux weakening control \cite{2016GaoPMSM}, and $i_q^r$ is obtained from the outer loop PI to track dc bus voltage.}  \blue{Once the inputs \eqref{eq:dddq} and the current references \eqref{eq:iqr} are substituted in the open loop system dynamics \eqref{eq:PMSMslow} (slow states) and \eqref{eq:PMSMfast} (fast states), the closed loop system becomes:
\begin{align}
    \dot{x} \s= f(x,\;z) \\
    \dot{z} \s = g(x,\;z) + Bv \label{eq:pmsmclosed_z}
\end{align}
where $x= \br{{v}_{dc},\;{\zeta}_{v_{dc}}}^T\in\mb{R}^2,\;z=\br{i_d,\;\zeta_{id},\;i_q,\;\zeta_{iq}}^T\in\mb{R}^4,\;v = \br{v_d,\;v_q}^T\in\mb{R}^2$, and the matrix $B\in\mb{R}^{4\times 2}$ is:
\begin{align}\renewcommand{\arraystretch}{1.15}\label{eq:Bpmsm}
    B = \begin{pmatrix}\cfrac{v_{dc}}{2L_d}&0\\0 \s 0\\ 0&\cfrac{v_{dc}}{2L_q}\\ 0\s 0\end{pmatrix}.
\end{align}}

\begin{figure}[!b]
	\begin{center}
		\centering
		\includegraphics[width=0.5\textwidth]{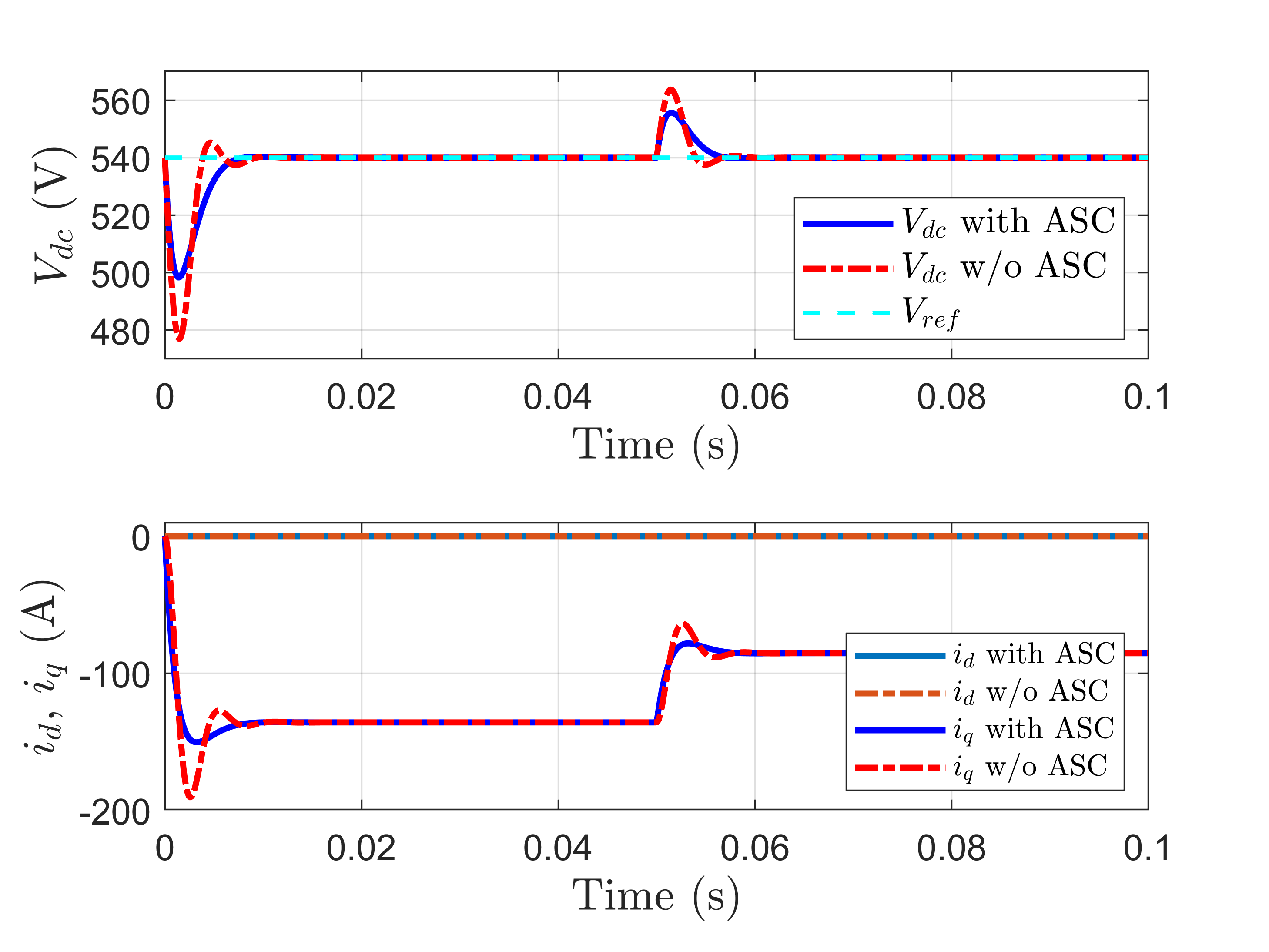}
		\caption{\label{fig:pmsm_plot} Simulation results for PMSM with active rectification. Top: Load voltage $(v_{dc})$, bottom: Currents in d-q frame $(i_{d})$ and $(i_{q})$.}
	\end{center}
\end{figure}
\blue{The goal is then to use the new   terms, $v = \br{v_d,\;v_q}^T$, for the sensitivity conditioning implementation:
\begin{align}
    Bv(t)&=S_x^z\big(x(t),\;z(t)\big)f\big(x(t),\;z(t)\big) 
\end{align}
}
\blue{
However, since the number of inputs is less than the dimension of the fast states, the matrix $B$ in \eqref{eq:Bpmsm} is not invertible. Therefore, approximate sensitivity conditioning is implemented based on \eqref{eq:hderiv}, \eqref{eq:pred_model}, and \eqref{eq:apprxsc}, as follows:
\begin{align}
    v \s= B^\dagger_LS^z_x(x,\;z)f(x,\;z) \\
    \Ra  v\s= -B^\dagger_L\pa{\nabla_z g(x,\;z)}^{-1}\nabla_x g(x,\;z)f(x,\;z)
\end{align}
where the  gradient terms  are obtained analytically from $g(x,\;z)$ in \eqref{eq:pmsmclosed_z}. }

The simulation parameters for this case are shown in Table \ref{tab:case3}, and Fig. \ref{fig:pmsm_controllers} summarizes the closed loop control with approximate sensitivity conditioning. The PMSM model is initialized with the load at the dc side drawing 33.48 kW at 540 V. At time $t = 0.05$ $s$, the power consumption is reduced to 18.9 kW. The simulation results are shown in Fig. \ref{fig:pmsm_plot}. As seen from this figure, the sensitivity conditioning strategy provides a better response in comparison to the traditional dual-PI cascade controller.

\begin{figure}[!b]\newcommand{\meas}{.27}
	\centering
	\subfloat[Hardware test-bed demonstrating the control of buck converter.]{\label{fig:testbed}\includegraphics[height=\meas\textwidth]{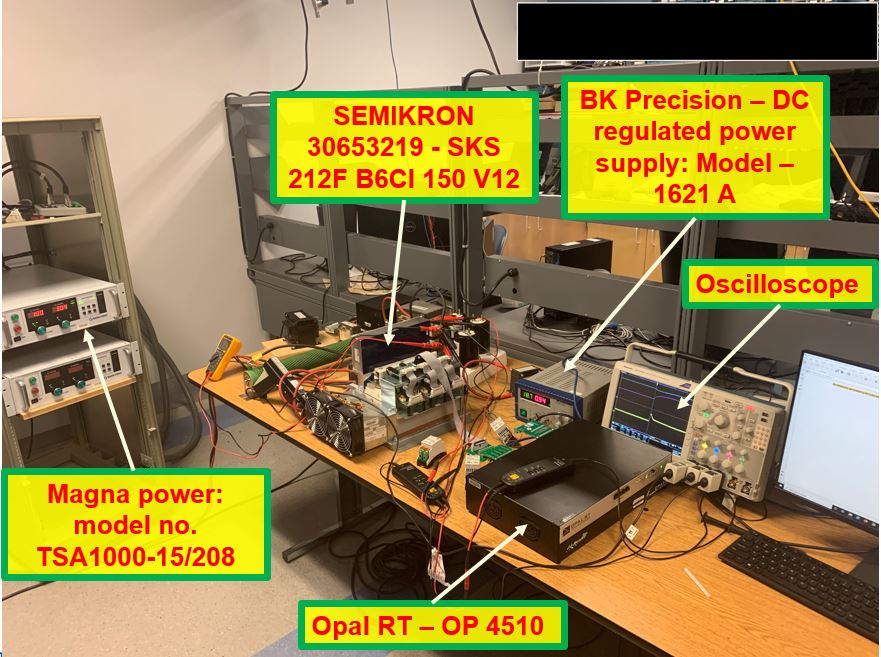}}\\
	\subfloat[High level schematic of buck converter interconnections.]{\label{fig:schematic}\includegraphics[height=\meas\textwidth]{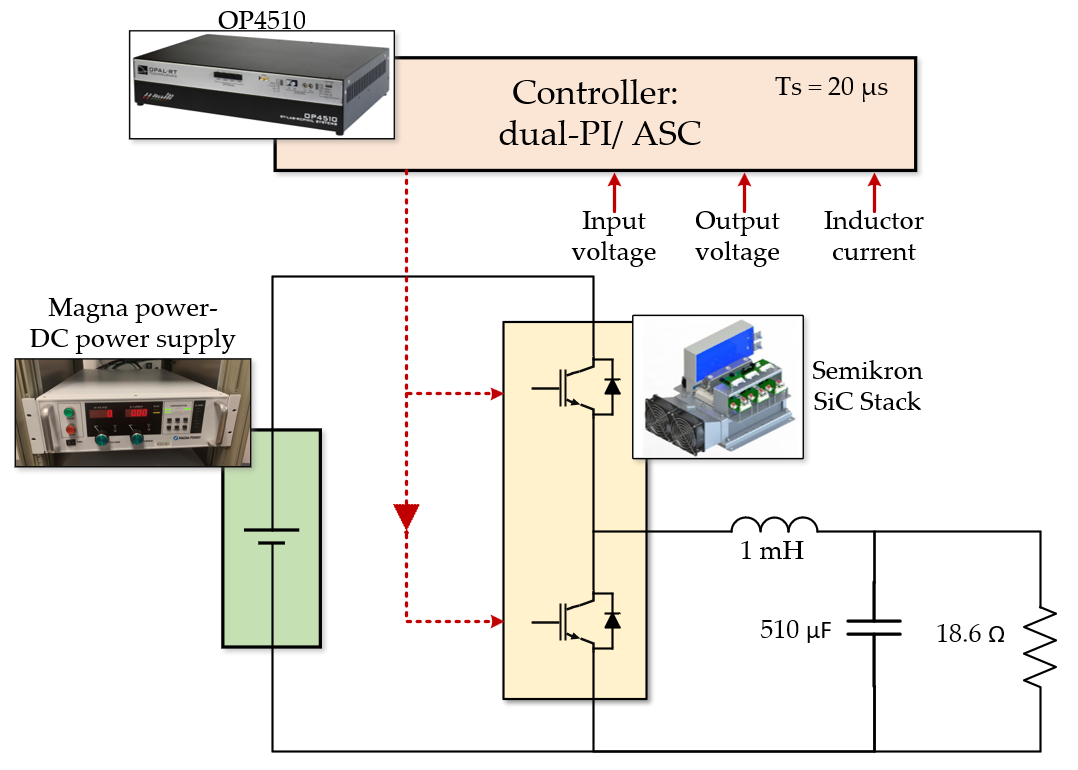}}\\
\caption{Experimental results: Hardware test-bed and interconnections.}\label{fig:hardware}
\end{figure}
\section{Experimental Results}
In this section, experimental results are shown for a closed loop buck converter, similar to case 1 in the previous section. The parameters of the test-bed and PI gains are the same as those in Table \ref{tab:case1}. An input voltage of 100 V  is provided by Magna Power's DC power supply \cite{magna}.  The switching devices are based on Semikron Silicon Carbide half-bridge module \cite{semikron}. The controller is developed using Opal RT OP4510 \cite{OpalRT}, containing the dual-PI cascade controller and the approximate sensitivity conditioning term. The hardware test-bed is shown in Fig. \ref{fig:testbed}, with a high-level diagram in Fig. \ref{fig:schematic}. The controller is implemented in OP4510 with a time step of $T_{sc} =20\;\mu\txt{s}$. 

\begin{figure}[!t]\newcommand{\meas}{.44}
	\centering
	\subfloat[Results with regular dual-PI cascade controller.]{\label{fig:exp_nps}\includegraphics[width=\meas\textwidth]{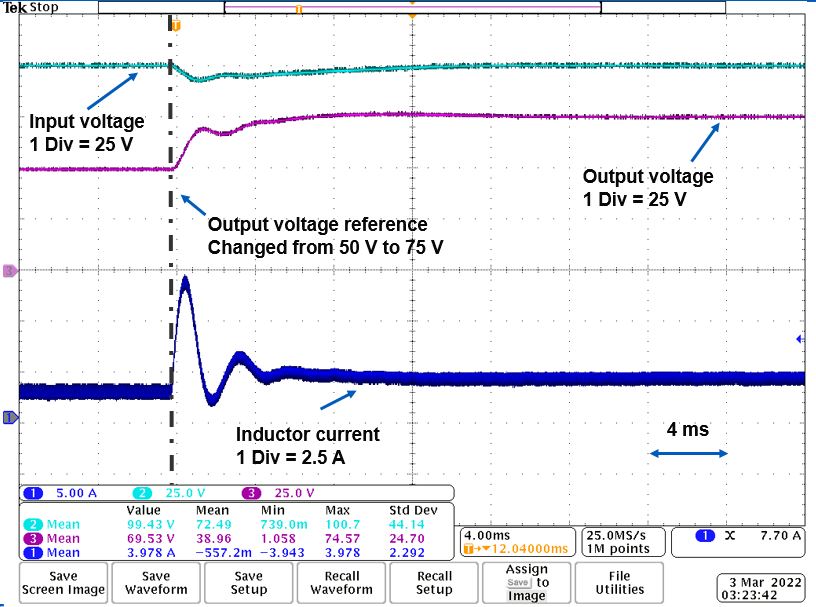}}\\
	\subfloat[Results with approximate sensitivity control.]{\label{fig:exp_ps}\includegraphics[width=\meas\textwidth]{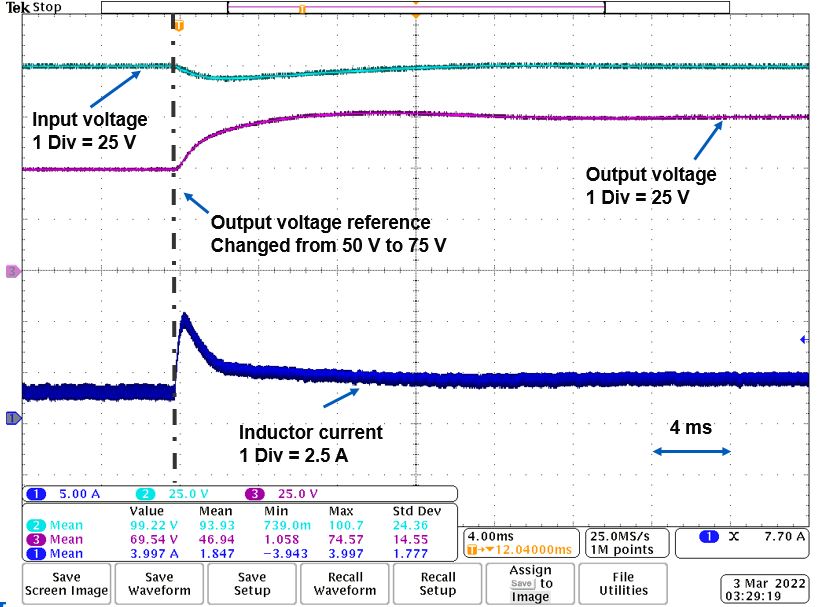}}\\
\caption{Experimental results without (top) and with (bottom) approximate sensitivity conditioning.}\label{fig:exp}
\end{figure}

The experimental results are shown in Fig. \ref{fig:exp}.  The inductor current is shown by the dark blue trace, the input voltage in light blue, and output voltage in purple.  The results with the typical cascaded control without sensitivity conditioning are shown in Fig. \ref{fig:exp_nps}, while the results with approximate sensitivity conditioning are shown in Fig. \ref{fig:exp_ps}.  As it can be seen, the approximate sensitivity conditioning helps in improving the transient response and the overall closed loop system. A comparison of these two results is shown in Fig. \ref{fig:exp_analysis}. It is emphasized that controller gains are the same for both cases, with the main difference being the   approximate sensitivity conditioning term, obtained using \eqref{eq:apprxsc}. It can be seen that the system with   approximate sensitivity conditioning exhibits a better transient response.

\begin{figure}[!t]
	\begin{center}
		\centering
		\includegraphics[width=0.5\textwidth]{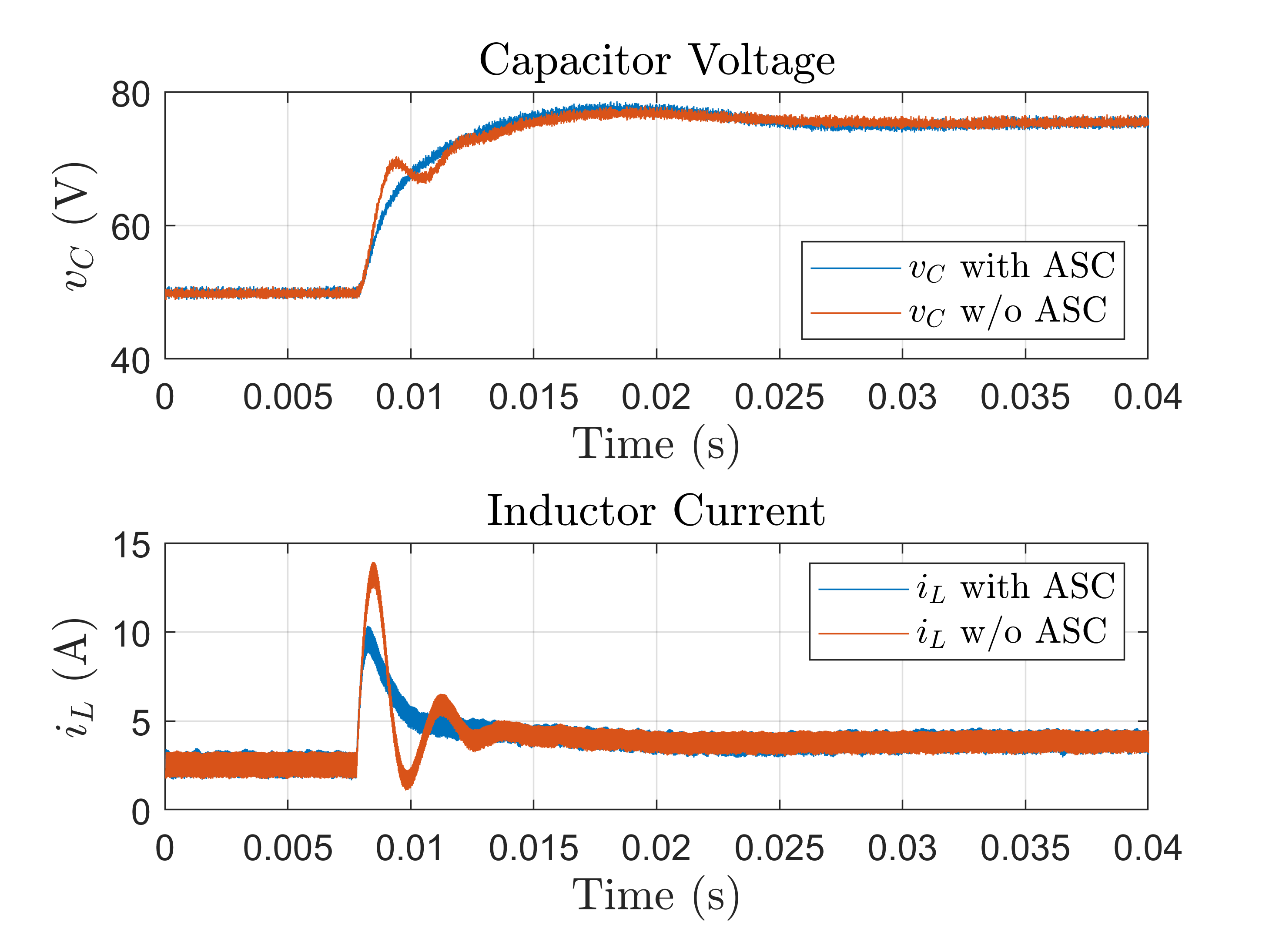}
		\caption{\label{fig:exp_analysis} Comparison of the experimental results for a buck converter with and without approximate sensitivity conditioning.}
	\end{center}
\end{figure}

\section{Conclusion and Future Work}
In this paper, an in-depth analysis of sensitivity conditioning   control \cite{picallo2021predictive} is presented using singular perturbation theory.  The impact of the sensitivity conditioning is discussed for systems with slow and fast states/dynamics and  it  is shown that as the time scale separation between these states grow, the impact of the sensitivity term is reduced.  In addition, the implementation of sensitivity conditioning for linear systems is discussed in detail, and different cases are presented based on the number of inputs.  For most practical applications, where the number of inputs is less than the number of fast states, an exact   sensitivity conditioning implementation is not possible.  For this case, an approximate sensitivity conditioning strategy is proposed, where the   term minimizes the least square error. 

Simulation results are presented for power converter applications with cascaded inner/outer loop PI control.  In the first case, a buck converter with output capacitor voltage regulation is presented and different gains are used to test the impact of   term at different time scale separations. Next, an actively rectified PMSM is studied, with nonlinear system dynamics.  A cascaded baseline controller using field oriented control is presented and the approximate sensitivity conditioning term is implemented. Lastly, experimental results are presented for a buck converter, demonstrating the advantage of sensitivity conditioning in improving the closed loop behavior with applications in power electronics.

For future work, the sensitivity conditioning control strategy will be analyzed for more complex and large scale nonlinear systems with applications in microgrids and power systems.   \blue{Lastly, the combination of high gain feedback with sensitivity conditioning will be studied with the aim of conducting a more comprehensive analysis and achieving better performance for systems with separate time scales.}

\bibliographystyle{IEEEtran}
\bibliography{IEEEabrv,bibfile_arc}

\begin{thebibliography}{10}
\providecommand{\url}[1]{#1}
\csname url@samestyle\endcsname
\providecommand{\newblock}{\relax}
\providecommand{\bibinfo}[2]{#2}
\providecommand{\BIBentrySTDinterwordspacing}{\spaceskip=0pt\relax}
\providecommand{\BIBentryALTinterwordstretchfactor}{4}
\providecommand{\BIBentryALTinterwordspacing}{\spaceskip=\fontdimen2\font plus
\BIBentryALTinterwordstretchfactor\fontdimen3\font minus
  \fontdimen4\font\relax}
\providecommand{\BIBforeignlanguage}[2]{{%
\expandafter\ifx\csname l@#1\endcsname\relax
\typeout{** WARNING: IEEEtran.bst: No hyphenation pattern has been}%
\typeout{** loaded for the language `#1'. Using the pattern for}%
\typeout{** the default language instead.}%
\else
\language=\csname l@#1\endcsname
\fi
#2}}
\providecommand{\BIBdecl}{\relax}
\BIBdecl

\bibitem{aluisio}
B.~Aluisio, S.~Bruno, L.~De~Bellis, M.~Dicorato, G.~Forte, and M.~Trovato,
  ``Dc-microgrid operation planning for an electric vehicle supply
  infrastructure,'' \emph{Applied Sciences}, vol.~9, no.~13, p. 2687, 2019.

\bibitem{aluisio2}
B.~Aluisio, M.~Dicorato, I.~Ferrini, G.~Forte, R.~Sbrizzai, and M.~Trovato,
  ``Optimal sizing procedure for electric vehicle supply infrastructure based
  on dc microgrid with station commitment,'' \emph{Energies}, vol.~12, no.~10,
  p. 1901, 2019.

\bibitem{giampaolo}
G.~{Buticchi}, S.~{Bozhko}, M.~{Liserre}, P.~{Wheeler}, and K.~{Al-Haddad},
  ``On-board microgrids for the more electric aircraft—technology review,''
  \emph{IEEE Transactions on Industrial Electronics}, vol.~66, no.~7, pp.
  5588--5599, July 2019.

\bibitem{xu2019}
Q.~Xu, Y.~Xu, P.~Tu, T.~Zhao, and P.~Wang, ``Systematic reliability modeling
  and evaluation for on-board power systems of more electric aircrafts,''
  \emph{IEEE Transactions on Power Systems}, vol.~34, no.~4, pp. 3264--3273,
  2019.

\bibitem{he2018}
J.~He, D.~Zhang, and D.~Torrey, ``Recent advances of power electronics
  applications in more electric aircrafts,'' in \emph{2018 AIAA/IEEE Electric
  Aircraft Technologies Symposium (EATS)}, 2018, pp. 1--8.

\bibitem{kim2018}
M.~Kim, S.~G. Lee, and S.~Bae, ``Decentralized power management for electrical
  power systems in more electric aircrafts,'' \emph{Electronics}, vol.~7,
  no.~9, p. 187, 2018.

\bibitem{shirin}
S.~{Yousefizadeh}, J.~D. {Bendtsen}, N.~{Vafamand}, M.~H. {Khooban},
  F.~{Blaabjerg}, and T.~{Dragičević}, ``Tracking control for a dc microgrid
  feeding uncertain loads in more electric aircraft: Adaptive backstepping
  approach,'' \emph{IEEE Transactions on Industrial Electronics}, vol.~66,
  no.~7, pp. 5644--5652, July 2019.

\bibitem{2017HerreraTSG}
L.~Herrera, W.~Zhang, and J.~Wang, ``Stability analysis and controller design
  of dc microgrids with constant power loads,'' \emph{IEEE Transactions on
  Smart Grid}, vol.~8, no.~2, pp. 881--888, 2017.

\bibitem{airan}
A.~{Francés-Roger}, A.~{Anvari-Moghaddam}, E.~{Rodríguez-Díaz}, J.~C.
  {Vasquez}, J.~M. {Guerrero}, and J.~{Uceda}, ``Dynamic assessment of cots
  converters-based dc integrated power systems in electric ships,'' \emph{IEEE
  Transactions on Industrial Informatics}, vol.~14, no.~12, pp. 5518--5529, Dec
  2018.

\bibitem{alfalahi}
M.~Al-Falahi, T.~Tarasiuk, S.~Jayasinghe, Z.~Jin, H.~Enshaei, and J.~Guerrero,
  ``Ac ship microgrids: control and power management optimization,''
  \emph{Energies}, vol.~11, no.~6, p. 1458, 2018.

\bibitem{ali}
A.~Haseltalab, M.~A. Botto, and R.~R. Negenborn, ``Model predictive dc voltage
  control for all-electric ships,'' \emph{Control Engineering Practice},
  vol.~90, pp. 133--147, 2019.

\bibitem{fang2019}
S.~Fang, Y.~Wang, B.~Gou, and Y.~Xu, ``Toward future green maritime
  transportation: An overview of seaport microgrids and all-electric ships,''
  \emph{IEEE Transactions on Vehicular Technology}, vol.~69, no.~1, pp.
  207--219, 2019.

\bibitem{Sulligoi}
Z.~Jin, G.~Sulligoi, R.~Cuzner, L.~Meng, J.~C. Vasquez, and J.~M. Guerrero,
  ``Next-generation shipboard dc power system: Introduction smart grid and dc
  microgrid technologies into maritime electrical netowrks,'' \emph{IEEE
  Electrification Magazine}, vol.~4, no.~2, pp. 45--57, 2016.

\bibitem{vu2017}
T.~V. Vu, D.~Gonsoulin, D.~Perkins, F.~Diaz, H.~Vahedi, and C.~S. Edrington,
  ``Predictive energy management for mvdc all-electric ships,'' in \emph{2017
  IEEE Electric Ship Technologies Symposium (ESTS)}, 2017, pp. 327--331.

\bibitem{Zubieta}
L.~E. Zubieta, ``Are microgrids the future of energy?: Dc microgrids from
  concept to demonstration to deployment,'' \emph{IEEE Electrification
  Magazine}, vol.~4, no.~2, pp. 37--44, 2016.

\bibitem{lvdc}
T.~Dragicevic, J.~C. Vasquez, J.~M. Guerrero, and D.~Skrlec, ``Advanced lvdc
  electrical power architectures and microgrids: A step toward a new generation
  of power distribution networks.'' \emph{IEEE Electrification Magazine},
  vol.~2, no.~1, pp. 54--65, 2014.

\bibitem{Shenai}
K.~Shenai, A.~Jhunjhunwala, and P.~Kaur, ``Electrifying india: Using solar dc
  microgrids,'' \emph{IEEE Power Electronics Magazine}, vol.~3, no.~4, pp.
  42--48, 2016.

\bibitem{Nordman}
B.~Nordman and K.~Christensen, ``Dc local power distribution: Technology,
  deployment, and pathways to success,'' \emph{IEEE Electrification Magazine},
  vol.~4, no.~2, pp. 29--36, 2016.

\bibitem{bellinaso2018cascade}
L.~V. Bellinaso, H.~H. Figueira, M.~F. Basquera, R.~P. Vieira, H.~A.
  Gr{\"u}ndling, and L.~Michels, ``Cascade control with adaptive voltage
  controller applied to photovoltaic boost converters,'' \emph{IEEE
  Transactions on Industry Applications}, vol.~55, no.~2, pp. 1903--1912, 2018.

\bibitem{liang2019model}
Y.~Liang, Z.~Liang, D.~Zhao, Y.~Huangfu, L.~Guo, and B.~Zhao, ``Model
  predictive control of interleaved dc-dc boost converter with current
  compensation,'' in \emph{2019 IEEE International Conference on Industrial
  Technology (ICIT)}.\hskip 1em plus 0.5em minus 0.4em\relax IEEE, 2019, pp.
  1701--1706.

\bibitem{pisano2008cascade}
A.~Pisano, A.~Davila, L.~Fridman, and E.~Usai, ``Cascade control of pm dc
  drives via second-order sliding-mode technique,'' \emph{IEEE Transactions on
  Industrial Electronics}, vol.~55, no.~11, pp. 3846--3854, 2008.

\bibitem{kokotovic}
P.~Kokotovi{\'c}, H.~K. Khalil, and J.~O'reilly, \emph{Singular perturbation
  methods in control: analysis and design}.\hskip 1em plus 0.5em minus
  0.4em\relax SIAM, 1999.

\bibitem{shen2019singular}
F.~Shen, P.~Ju, M.~Shahidehpour, Z.~Li, C.~Wang, and X.~Shi, ``Singular
  perturbation for the dynamic modeling of integrated energy systems,''
  \emph{IEEE Transactions on Power Systems}, vol.~35, no.~3, pp. 1718--1728,
  2019.

\bibitem{mchaouar2018nonlinear}
Y.~Mchaouar, A.~Abouloifa, I.~Lachkar, M.~Fettach, F.~Giri, C.~Taghzaoui, and
  A.~Elallali, ``Nonlinear control of single-phase shunt active power filters
  with theoretical analysis via singular perturbation,'' in \emph{2018 5th
  International Conference on Electrical and Electronic Engineering
  (ICEEE)}.\hskip 1em plus 0.5em minus 0.4em\relax IEEE, 2018, pp. 67--72.

\bibitem{mchaouar2020new}
Y.~Mchaouar, A.~Abouloifa, I.~Lachkar, and M.~Fettach, ``A new control strategy
  for photovoltaic system connected to the grid via three-time-scale singular
  perturbation technique with performance analysis,'' \emph{Advanced
  Statistical Modeling, Forecasting, and Fault Detection in Renewable Energy
  Systems}, p. 159, 2020.

\bibitem{gui2021large}
Y.~Gui, R.~Han, J.~M. Guerrero, J.~C. Vasquez, B.~Wei, and W.~Kim,
  ``Large-signal stability improvement of dc-dc converters in dc microgrid,''
  \emph{IEEE Transactions on Energy Conversion}, vol.~36, no.~3, pp.
  2534--2544, 2021.

\bibitem{liu2021existence}
Z.~Liu, R.~Liu, Z.~Xia, M.~Su, X.~Deng, X.~Zhang, and J.~Lu, ``Existence and
  stability of equilibrium of dc micro-grid under master-slave control,''
  \emph{IEEE Transactions on Power Systems}, vol.~37, no.~1, pp. 212--223,
  2021.

\bibitem{perez2018dc}
F.~Perez, A.~Iovine, G.~Damm, and P.~Ribeiro, ``Dc microgrid voltage stability
  by dynamic feedback linearization,'' in \emph{2018 IEEE International
  Conference on Industrial Technology (ICIT)}.\hskip 1em plus 0.5em minus
  0.4em\relax IEEE, 2018, pp. 129--134.

\bibitem{bayardo2020adaptive}
R.~G. Bayardo, A.~G. Loukianov, R.~Q. Fuentes-Aguilar, and V.~I. Utkin,
  ``Adaptive speed tracking controller for a brush-less dc motor using singular
  perturbation.'' \emph{IFAC-PapersOnLine}, vol.~53, no.~2, pp. 3880--3885,
  2020.

\bibitem{rasheduzzaman2015reduced}
M.~Rasheduzzaman, J.~A. Mueller, and J.~W. Kimball, ``Reduced-order
  small-signal model of microgrid systems,'' \emph{IEEE Transactions on
  Sustainable Energy}, vol.~6, no.~4, pp. 1292--1305, 2015.

\bibitem{mariani2014model}
V.~Mariani, F.~Vasca, J.~C. Vasquez, and J.~M. Guerrero, ``Model order
  reductions for stability analysis of islanded microgrids with droop
  control,'' \emph{IEEE Transactions on Industrial Electronics}, vol.~62,
  no.~7, pp. 4344--4354, 2014.

\bibitem{kimball2008singular}
J.~W. Kimball and P.~T. Krein, ``Singular perturbation theory for dc--dc
  converters and application to pfc converters,'' \emph{IEEE Transactions on
  Power Electronics}, vol.~23, no.~6, pp. 2970--2981, 2008.

\bibitem{picallo2021predictive}
M.~Picallo, S.~Bolognani, and F.~Dorfler, ``Sensitivity-conditioning: Beyond
  singular perturbation for control design on multiple time scales,''
  \emph{IEEE Transactions on Automatic Control, Early Access}, pp. 1--16, 2022.

\bibitem{simonetto2016class}
A.~Simonetto, A.~Mokhtari, A.~Koppel, G.~Leus, and A.~Ribeiro, ``A class of
  prediction-correction methods for time-varying convex optimization,''
  \emph{IEEE Transactions on Signal Processing}, vol.~64, no.~17, pp.
  4576--4591, 2016.

\bibitem{fazlyab2017prediction}
M.~Fazlyab, S.~Paternain, V.~M. Preciado, and A.~Ribeiro,
  ``Prediction-correction interior-point method for time-varying convex
  optimization,'' \emph{IEEE Transactions on Automatic Control}, vol.~63,
  no.~7, pp. 1973--1986, 2017.

\bibitem{mojica2021stackelberg}
E.~Mojica-Nava and F.~Ruiz, ``Stackelberg population dynamics: A
  predictive-sensitivity approach,'' \emph{Games}, vol.~12, no.~4, p.~88, 2021.

\bibitem{hkkhalil}
H.~K. Khalil, ``Nonlinear systems third edition,'' \emph{Patience Hall}, vol.
  115, 2002.

\bibitem{o1991singular}
R.~E. O'Malley, \emph{Singular perturbation methods for ordinary differential
  equations}.\hskip 1em plus 0.5em minus 0.4em\relax Springer, 1991, vol.~89.

\bibitem{1977YoungHGF}
K.-K. Young, P.~Kokotovic, and V.~Utkin, ``A singular perturbation analysis of
  high-gain feedback systems,'' \emph{IEEE Transactions on Automatic Control},
  vol.~22, no.~6, pp. 931--938, 1977.

\bibitem{1987KhalilHGF}
H.~Khalil and A.~Saberi, ``Adaptive stabilization of a class of nonlinear
  systems using high-gain feedback,'' \emph{IEEE Transactions on Automatic
  Control}, vol.~32, no.~11, pp. 1031--1035, 1987.

\bibitem{krantz2002implicit}
S.~G. Krantz and H.~R. Parks, \emph{The implicit function theorem: history,
  theory, and applications}.\hskip 1em plus 0.5em minus 0.4em\relax Springer
  Science \& Business Media, 2002.

\bibitem{barata2012moore}
J.~C.~A. Barata and M.~S. Hussein, ``The moore--penrose pseudoinverse: A
  tutorial review of the theory,'' \emph{Brazilian Journal of Physics},
  vol.~42, no.~1, pp. 146--165, 2012.

\bibitem{herrera2021nonlinear}
L.~Herrera, C.~Miller, and B.-H. Tsao, ``Nonlinear model predictive control of
  permanent magnet synchronous generators in dc microgrids,'' in \emph{2021
  American Control Conference (ACC)}.\hskip 1em plus 0.5em minus 0.4em\relax
  IEEE, 2021, pp. 5049--5054.

\bibitem{2016GaoPMSM}
F.~Gao and S.~Bozhko, ``Modeling and impedance analysis of a single dc
  bus-based multiple-source multiple-load electrical power system,'' \emph{IEEE
  Transactions on Transportation Electrification}, vol.~2, no.~3, pp. 335--346,
  2016.

\bibitem{magna}
``{Magna-Power Programmable DC Power Supply},''
  \url{https://magna-power.com/products/magnadc/ts}.

\bibitem{semikron}
``{Semikron},'' \url{https://www.semikron.com/}.

\bibitem{OpalRT}
``{Opal-RT Technologies},''
  \url{https://www.opal-rt.com/simulator-platform-op4510/}.

\end{thebibliography}

\end{document}